\documentclass[11pt]{amsart}
\usepackage{booktabs,tabularx,array,makecell,threeparttable}
\usepackage{amsmath,amssymb,amsthm,graphicx,bbm}
\usepackage[left=1in,right=1in,top=1in, bottom=1in]{geometry}
\usepackage{cite}
\usepackage[dvipsnames]{xcolor}
\usepackage[british]{babel}
\usepackage{hyperref}
\usepackage{enumitem}
\usepackage{mathtools}
\usepackage{comment}
\usepackage{subcaption}
\usepackage{lipsum}
\usepackage{algorithm}
\usepackage{algpseudocode}
\usepackage{multirow}
\usepackage{xcolor}
\makeatletter
\newcommand{\customlabel}[2]{#2\def\@currentlabel{#2}\label{#1}}
\makeatother
\hypersetup{
    colorlinks=false,
    pdfborder={0 0 0},
}
\usepackage{tikz}
\theoremstyle{plain}
\newtheorem{thm}{Theorem}[section]
\newtheorem{lem}[thm]{Lemma}
\newtheorem{prop}[thm]{Proposition}
\newtheorem{cor}[thm]{Corollary}

\newtheorem*{assumption*}{Assumption}
\newtheorem{example}{Example}[section]

\theoremstyle{remark}
\newtheorem{rem}[thm]{Remark}

\newenvironment{namedalgorithm}[2]
{
\begin{algorithm}[H]
\renewcommand{\thealgorithm}{\textcolor{blue}{(#1)}}
\caption{#2}\label{alg:#1}
}
{
\end{algorithm}
}

\DeclareMathOperator{\E}{\mathbb E}
\renewcommand{\P}{\mathbb P}

\renewcommand{\epsilon}{\varepsilon}

\newcommand{\TV}{\mathrm{TV}}

\newcommand{\eps}{\varepsilon}

\newcommand{\ud}{{\mathrm d}}
\newcommand{\CP}{{\mathrm{RC}}}
\newcommand{\R}{{\mathbb R}}

\newcommand{\N}{{\mathbb N}}

\newcommand{\Id}{\text{Id}}

\newcommand{\RCtot}{\mathrm{RC}_\mathrm{total}}
\makeatletter
\def\namedlabel#1#2{\begingroup
    #2%
    \def\@currentlabel{#2}%
    \phantomsection\label{#1}\endgroup
}
\makeatother

\begin{document}


\title[Diffeomorphic Monte Carlo]{Diffeomorphic Markov Chain Monte Carlo: fast mixing \\for heavy-tailed distributions}
\author{Miha Bre\v{s}ar}
\address{School of Data Science, The Chinese University of Hong Kong, Shenzhen, China} 
\email{mihabresar@cuhk.edu.cn}
\author{Aleksandar Mijatovi\'c}
\address{Department of Statistics, University of Warwick, UK}
\email{a.mijatovic@warwick.ac.uk}



\begin{abstract}
We introduce a new class of uniformly ergodic MCMC algorithms, termed  Diffeomorphic Contraction Sampler (DCS), and provide fast non-asymptotic mixing guarantees for DCS targeting distributions on $\R^d$ with arbitrarily heavy polynomial tails. DCS provides a solution to a well-known problem for  MCMC samplers, which typically struggle with the combination of unbounded high-dimensional state space and vanishing gradients.

The DCS pulls back a target on $\R^d$ onto a Euclidean ball $B(R)\subset\R^d$ and then samples from the transformed density on the convex set $B(R)$ via algorithms such as the Ball Walk, Hit-and-Run and others. A radial diffeomorphic contraction is chosen so that the  pull-back density on $B(R)$ is bounded, implying uniform ergodicity for \textit{all} targets with a finite polynomial moment. Non-asymptotic bounds for DCS require stronger assumptions such as log-concavity of the pull-back density. In practice, this is achieved approximately by a preconditioned  automorphism of the ball $B(R)$, tuned via Variational Inference. 

Numerical simulation tests demonstrate that the DCS outperforms significantly  the No-U-Turns sampler  on multi-dimensional heavy-tailed targets arising as real-world posteriors in PosteriorDB~\cite{posteriordb} benchmark.  DCS also numerically outperforms in high-dimensional examples recently developed spherical projection samplers~\cite{grazzi2026sub,Yang24} for heavy-tailed target distributions.
\end{abstract}

\maketitle

\section{Introduction}

Sampling from heavy-tailed target distributions in high dimensions is a notoriously difficult problem due to the compounding effects of unbounded state spaces and vanishing gradients in the tails~\cite{Mengersen96,brevsar2025central,Roberts07}. In such settings, Markov chain Monte Carlo (MCMC) methods typically exhibit excessively slow convergence and the central limit theorem for ergodic average estimators fails frequently~\cite{brevsar2025central}. A promising direction for improving the performance of vanilla samplers in this context relies on density transformations, which have been shown to yield faster asymptotic rates of convergence~\cite{grazzi2026sub,MR3097969,Yang24}. This paper introduces the \textit{diffeomorphic contraction sampler}~\ref{alg:DCS}, based on a  transformation via a diffeomorphic radial contraction that pulls back the heavy-tailed target on the Euclidean space $\R^d$ into a density on the ball $B(R)$ of radius $R\in(0,\infty)$ in $\R^d$. This diffeomorphism translates the simulation problem into sampling from a distribution on a cornerless convex body $B(R)$, where efficient algorithms, such as the Ball Walk~\ref{alg:DCS-BW} and Hit-and-Run~\ref{alg:DCS-HnR}, are known to work well. 
The choice of the sampling domain $B(R)$ is not arbitrary:  
sampling on a ball has the best mixing time bounds among all convex bodies in $\R^d$, see~\cite{MR2309621,jiang2024regularized}. 

The following features \textbf{(I)-(III)} of Algorithm~\ref{alg:DCS}  are  the main  contributions of the paper.\\
\noindent \textbf{(I) Uniform ergodicity of Algorithm~\ref{alg:DCS} for all polynomial targets.} We prove that Algorithm~\ref{alg:DCS} 
achieves uniform ergodicity for arbitrarily heavy polynomial tails of the target (see Theorem~\ref{thm:Uniform_ergodicity_unified} in Section~\ref{subsec:uniform_ergodicity} below). It is known that vanilla MCMC algorithms are only polynomially ergodic for polynomial targets~\cite{brevsar2025central,Roberts07,livingstone2019geometric}.
Moreover,~\ref{alg:DCS} is a strict improvement on  recent transformation-based methods~\cite{Yang24,grazzi2026sub} that require certain tail decay of the target for the algorithm to exhibit uniform ergodicity. If these tail-decay conditions are violated, the algorithms in~\cite{Yang24,grazzi2026sub} are only polynomially ergodic (see~\cite{brevsar2025central} and numerical example in Section~\ref{subsubsec:infinite_first_moment} below). \\
\noindent \textbf{(II) Non-asymptotic convergence bounds for Algorithm~\ref{alg:DCS}.} We provide  rigorous non-asymptotic mixing time bounds for Algorithm~\ref{alg:DCS} under appropriate assumptions on the target density (see Theorem~\ref{thm:dcs_bounds} in Section~\ref{subsec:non_asymptotic} below). 
The key feature of Algorithm~\ref{alg:DCS} is that the sampling  on bounded compact sets is in this case applied to the ball $B(R)$ in $\R^d$, where random walk Metropolis type algorithms (such as the Ball Walk~\ref{alg:DCS-BW}) typically do not get stuck as there are no corners in the state space~\cite{MR2309621,jiang2024regularized}. 

Existing non-asymptotic mixing-time results for heavy-tailed sampling are to the best of our knowledge limited to the Random Walk Metropolis (RWM) with Gaussian proposals~\cite{MR5033571}. From a cold start, the mixing-time bound for the RWM exhibits an exponential dependence $\exp(d)$ on the dimension $d$ for a symmetric Student-$t$ target~\cite[Example~92]{MR5033571}. Moreover, the lower bounds on the convergence rate in~\cite[Sec~3.1]{brevsar2025central} for RWM suggest optimality for the non-asymptotic result in~\cite[Example~92]{MR5033571}. 
In contrast, Algorithm~\ref{alg:DCS} (based on Hit-and-Run~\ref{alg:DCS-HnR}) has a polynomial mixing time proportional to $d^3$ from a cold start for symmetric Student-$t$ targets (see Theorem~\ref{thm:dcs_bounds} and Example~\ref{ex:student_mixing} below).\\
\noindent \textbf{(III) Robustness and numerical efficiency of Algorithm~\ref{alg:DCS}.} In Section~\ref{sec:sim} we conduct extensive testing of the performance  of Algorithm~\ref{alg:DCS} using both real-world and hybrid data. These tests demonstrate the robustness of~\ref{alg:DCS} to arbitrary polynomial tail decay of the target and complex geometry of the landscape induced by the potential. In particular,  Algorithm~\ref{alg:DCS} significantly numerically outperforms JAX hardware-accelerated No-U-Turns sampler implementation~\cite{hoffman2014no, phan2019composable, jax2018github}. This performance advantage holds for real-world PosteriorDB benchmarks~\cite{posteriordb, magnusson2024posteriordb}, simultaneously exhibiting funnel geometries and heavy tails~\cite{neal2003slice} (Section~\ref{subsec:data_real}), and for skewed targets with polynomial tails (Section~\ref{subsec:art_data}).
 In the numerical examples of Section~\ref{subsec:comp_stereo}, Algorithm~\ref{alg:DCS} outperforms state-of-the-art samplers~\cite{Yang24,grazzi2026sub} designed for heavy-tailed target distributions in the regime when the sampler in~\cite{grazzi2026sub} is uniformly ergodic (and, less surprisingly, also when it is not). 

Fast non-asymptotic mixing~\textbf{(II)} requires suitable properties (near isotropy/log-concavity on the ball $B(R)$) of the pull-back measure under the radial diffeomorphism. Inspired by adaptive MCMC methods~\cite{hoffman2014no,bell2024adaptive,MR3671776} and the preconditioning of sampling spaces via transport maps~\cite{parno2018transport,hoffman2019neutra}, we introduce a Variational Inference (VI) training phase in Algorithm~\ref{alg:DCS} to optimize the transformation parameters of the automorphism of $B(R)$ (see Section~\ref{sec:optimization} below). 
We stress that the tail-condition for uniform ergodicity~\textbf{(I)} (see Assumption~\eqref{eq:unif_ergodicity_pi_condition} below) is not affected by this step.
The preprocessing  smooths complicated landscape geometries into near-uniform target densities that are proven to mix fast (see results in Section~\ref{subsec:non_asymptotic}). Put differently, the radial-contraction diffeomorphism makes sure that the sampling algorithm is uniformly ergodic~\textbf{(I)}, which deals with the non-log-concavity in the tails of the target. The VI preprocessing ensures good non-asymptotic performance~\textbf{(II)} of Algorithm~\ref{alg:DCS}  for non-log-concave (i.e., multimodal)  target distributions.

The remainder of the paper is organised as follows. Section~\ref{sec:radial_contraction} defines Algorithm~\ref{alg:DCS}. Section~\ref{sec:theoretical_results} formulates and discusses uniform ergodicity and the non-asymptotic mixing-time bounds for Algorithm~\ref{alg:DCS}. Subsection~\ref{subsec:literature} describes the related literature for heavy-tailed sampling algorithms. Section~\ref{sec:sim} provides numerical evidence of the efficiency of Algorithm~\ref{alg:DCS} for heavy-tailed targets. Section~\ref{sec:conclusions} concludes the paper and states open problems in the context of Algorithm~\ref{alg:DCS}. Appendix~\ref{app:proofs_non_asymptotic} provides proofs for our theoretical results. Appendix~\ref{app:diffeo_convex} verifies the  assumptions for non-asymptotic mixing-time bounds  for a family of target distributions on $\R^d$, which includes Student-$t$ densities. 
Finally, Appendices~\ref{app:Algs} and~\ref{app:mobius_vi} provide, respectively, examples of algorithms for sampling on the ball $B(R)$ and details of the variational-inference procedure used in Algorithm~\ref{alg:DCS}. A short YouTube presentation~\cite{YouTube_talk} describes our \href{https://youtu.be/cOlmRLxeQM0}{\underline{results}}.

\section{Diffeomorphic Markov Chain Monte Carlo Framework}
\label{sec:radial_contraction}

To efficiently sample from an arbitrary (possibly heavy-tailed)  distribution on $\R^d$, we construct a diffeomorphism $\CP_{\text{total}}: B(R) \to \R^d$ that maps a bounded sampling domain $B(R)$ (open ball of radius $R$) to the unbounded target space. We build this map in two stages: (I) a canonical radial contraction that handles the heavy (polynomial) tails is applied (Subsection~\ref{subsec:contraction}); (II) the contraction is composed with an  automorphism on $B(R)$ that re-centres and adjusts the target geometry (Subsection~\ref{subsec:automorphism}). Correctly tuned, Step~(I) ensures uniform ergodicity for arbitrary heavy-tailed targets. Step~(II) enables non-asymptotic analysis of the convergence and leads to rapid mixing in applications.  

\subsection{Diffeomorphic Radial Contraction}
\label{subsec:contraction}
Fix a state space $\R^d$, denote by $|x|$ the standard Euclidean norm on $\R^d$ for any $x\in\R^d$.  
Let $B(R)\coloneqq \{x\in\R^d:|x|<R\}$ denote the ball centred around the origin with radius $R>0$. The foundational component of our framework, which we now define, is the diffeomorphism $\CP_\beta: B(1)\to \R^d$ (for some $\beta\in(0,\infty)$) whose inverse acts as a non-uniform diffeomorphic radial contraction.

Let $\phi_\beta:[0,1)\to[0,\infty)$ be strictly increasing and differentiable, 
satisfying 
\begin{equation}
\label{eq:beta_condition_for_phi}
(1-t)\phi_\beta(t)^\beta \to c\quad\&\quad (1-t)\frac{\phi_\beta'(t)}{\phi_\beta(t)}\to c' \qquad\text{as $t\uparrow1$ for some $c,c'\in(0,\infty)$.}  
\end{equation}
Assume  $\CP_\beta(y) := \phi_\beta(|y|) y$ is a diffeomorphism with Jacobian
$J_{\CP_\beta}(y)\coloneqq |\det D(\CP_\beta)(y)|$
equal to $\phi_\beta\left(|y|\right)^{d-1}
\left(
\phi_\beta(|y|\right)
+\phi_\beta'\left(|y|\right)|y|)$,
$y \in B(1)$. 
It is easily seen that~\eqref{eq:beta_condition_for_phi} implies the following limit:
\begin{equation}
\label{eq:Jacobian_asymptotic_beta}\frac{J_{\CP_\beta}(y)}{|\CP_\beta(y)|^{d+\beta}}\to\frac{c'}{c}\in(0,\infty)\qquad\text{as $|y|\uparrow 1$.}
\end{equation}
The  motivation for our choice of the diffeomorphism $\CP_\beta$ is that the singularity of the Jacobian $J_{\CP_\beta}$ at the boundary of the ball $B(1)$ is, by the limit in~\eqref{eq:Jacobian_asymptotic_beta},
of order $(d+\beta)$. By picking $\beta\in(0,\infty)$ smaller than the 
polynomial rate of decay of the target distribution, limit~\eqref{eq:Jacobian_asymptotic_beta} implies the uniform ergodicity of Algorithm~\ref{alg:DCS} (see Theorem~\ref{thm:Uniform_ergodicity_unified} below  and Section~\ref{subsec:literature} for further discussion of the motivation and advantages of this transformation).



A simple choice for $\phi_\beta$ satisfying~\eqref{eq:beta_condition_for_phi}, used in all our examples (see Section~\ref{sec:sim} below), is given by
$$\phi_{\beta}(t)\coloneqq (1-t^\beta)^{-1/\beta}\quad \text{ for $t\in[0,1)$ and some  $\beta\in(0,\infty)$.}
$$
In this case the diffeomorphism
$\CP_\beta:B(1)\to\R^d$ is given by the formula
\begin{equation}
\label{eq:compact_projection_first}
    \CP_\beta(y) = \frac{y}{  \left(1 - |y|^\beta\right)^{1/\beta}} \quad \text{with Jacobian} \quad J_{\CP_\beta}(y) = \left( 1 - |y|^\beta \right)^{-(1 + d/\beta)},\quad y \in B(1).
\end{equation}


\subsection{Automorphism of the ball and the total diffeomorphism}
\label{subsec:automorphism}
While the base radial contraction provides a robust mechanism for handling heavy tails, its strict symmetry implicitly assumes that the probability mass is nearly isotropic and centred around the origin. For highly asymmetric targets, or targets with modes located far from the origin, a purely symmetric radial expansion is inefficient. Contracting a target distribution with modes far from the origin, into the ball may map most of its mass into regions of extremely high curvature of the Jacobian $J_{\CP_\beta}$ near the boundary. This would lead to poor mixing and ill-conditioned geometries for  sampling algorithms.

To address this issue, we introduce composition of diffeomorphisms 
$\CP_{\text{total}} : B(R) \to \R^d$, parametrised by $\theta = (\mu, R, \omega)$,
that geometrically pre-conditions the pull-back of the target distribution on the state space $B(R)$. To maintain analytical simplicity while preserving the physical sampling domain as $B(R)$, we decouple the geometric scale $R$ from the non-linear transformation of the state space given on the unit ball $B(1)$. 
More precisely, for any $z \in B(R)$, we normalize it to the unit ball $z/R \in B(1)$, apply a geometry-altering   diffeomorphism $\mathcal{T}_\omega: B(1) \to B(1)$ parametrised by $\omega$, map it to the unbounded space $\R^d$ using the canonical radial expansion $\CP_\beta:B(1)\to\R^d$ and finally apply rescaling and recentring in $\R^d$:
\begin{equation}
\label{eq:total_projection}
\tag{$\RCtot$}
    \CP_{\text{total}}(z) := \mu + R \cdot \CP_\beta(\mathcal{T}_\omega(z/R)),
    \qquad z\in B(R).
\end{equation}
\refstepcounter{equation}
The Jacobian $J_{\CP_{\text{total}}}(z)=|\det D(\CP_{\mathrm{total}})(z)|$ of $\CP_{\mathrm{total}}$ equals
\begin{equation}
\label{eq:total_jacobian}
J_{\CP_{\text{total}}}(z) = 
J_{\CP_\beta}\left(\mathcal{T}_\omega(z/R)\right)
J_{\mathcal{T}_\omega}(z/R)
\qquad\text{for  $z\in B(R)$.}
\end{equation}
The inverse of $\CP_{\text{total}}$ at  $x=\CP_{\text{total}}(z)$ is a diffeomorphism  $\CP_{\text{total}}^{-1}:\R^d\to B(R)$ given by 
\begin{equation}
\label{eq:total_inverse}
    \CP_{\text{total}}^{-1}(x) = R \cdot \mathcal{T}_\omega^{-1}\circ \CP_\beta^{-1}((x-\mu)/R), 
\end{equation}
If the diffeomorphism $\CP_\beta$ is given by~\eqref{eq:compact_projection_first}, its inverse equals 
 $\CP_\beta^{-1}(x')=x'/(1+|x'|^\beta)^{1/\beta}$, $ x'\in \R^d$.

Crucially, to preserve the boundary behaviour required by uniform ergodicity (Section~\ref{subsec:uniform_ergodicity} below), we assume $\mathcal{T}_\omega$ extends to a diffeomorphism on the closed ball $\overline{B}(1):=\{x\in\R^d: |x|\leq1\}$. This requirement guarantees that the Jacobian  $J_{\mathcal{T}_\omega}(y)$ is strictly bounded away from zero and infinity as $|y|\to1$, leaving the asymptotic tail mappings strictly intact. 

\subsection{Diffeomorphic Contraction Sampler}
Let $\pi:\R^d\to(0,\infty)$ be proportional to a probability density function of a measure $\nu$ on the Borel sets of $\R^d$. Consider the pull-back measure $\nu_B:=\nu(\CP_{\text{total}}(\cdot))$ on Borel sets in $B(R)$ (see~Proposition~\ref{prop:diffeo_invariance} below for formal definition). Let $\pi_B:B(R)\to(0,\infty)$ be a function proportional to the density $\nu_B:B(R)\to(0,\infty)$ of the probability measure 
$\nu_B$ (the symbol $\nu_B$ is overloaded throughout). The change-of-variable formula implies that  
\begin{equation}
\label{eq:total_target}
    \pi_{B}(z) := \pi(\CP_{\text{total}}(z)) \cdot  J_{\CP_{\text{total}}}(z), \quad z \in B(R),
\end{equation}
is proportional to the density of $\nu_B$ on $B(R)$.
If the density of $\nu$ satisfies the tail condition $\limsup_{|x|\to\infty}\pi(x)|x|^{d+\beta}<\infty$, the function $\pi_{B}$ (and hence the density of $\nu_B$) is bounded above by a constant on $B(R)$.

The core idea behind the framework is very simple: rather than exploring the heavy-tailed target $\pi$ on $\R^d$, we employ an efficient sampling algorithm (based on a Markov kernel $Q:B(R)\times\mathcal{B}(B(R))\to[0,1]$) to target the well-conditioned density $\pi_{B}$ on the bounded state space $B(R)$. The generated samples are then pushed forward via $\CP_{\text{total}}$ to obtain samples from the target distribution $\nu$. 

\begin{namedalgorithm}{DCS}{Diffeomorphic Contraction Sampler}
\begin{algorithmic}[1]
\Statex \textbf{Input:} $\pi$ proportional to density on $\R^d$, the map  in~\eqref{eq:total_projection}, Markov kernel $Q$ targeting $\nu_B$.
\State Current state: $X_n = x\in\R^d$
\State Simulate the increment $\hat{\mathbf{Z}}\in B(R)$ of the $Q$-chain:\label{alg:DCS:step2}
\begin{itemize}
    \item Map to the ball: $\mathbf{z} := \CP_{\text{total}}^{-1}(x)$ \Comment{Current state in $B(R)$ via~\eqref{eq:total_inverse}}
    \item Sample $\hat{\mathbf{Z}}$ from the Markov kernel $Q(\mathbf{z}, \cdot)$
\end{itemize}
\State Set $X_{n+1} = \CP_{\text{total}}(\hat{\mathbf{Z}})$
\end{algorithmic}
\end{namedalgorithm}

The Markov kernel $Q$ on $B(R)$ in Algorithm~\ref{alg:DCS} will be given by the widely used algorithms for sampling on convex bodies in~$\R^d$: Ball walk~\ref{alg:DCS-BW}, Hit-and-Run~\ref{alg:DCS-HnR}, Random Walk Metropolis~\ref{alg:DCS-RWM}, Langevin Monte Carlo~\ref{alg:DCS-LMC}  and Dikin walk~\ref{alg:DCS-DW} (see Appendix~\ref{app:Algs} below).  

\subsubsection{Hyperparameter Optimization in Algorithm~\ref{alg:DCS} via Variational Inference}
\label{sec:optimization}

In the simplest setting, when the target distribution $\pi$ is approximately centred and isotropic, Algorithm~\ref{alg:DCS} is effective with 
$\CP_\text{total}$ equal to the canonical radial contraction $\CP_{\beta}$ in~\eqref{eq:compact_projection_first} with a trivial  automorphism (e.g., $\mathcal{T}_\omega = \mathrm{Id}$) and without centring (i.e., $\mu=0$). In practice, however, target distributions are rarely symmetric: they may be anisotropic and/or exhibit more complicated geometries. To obtain the fast-mixing behaviour guaranteed by our non-asymptotic theory, we thus need to learn a suitable geometry-correcting diffeomorphism of the ball, $\mathcal{T}_\omega:B(1)\to B(1)$, together with the overall scale parameter $R$ and the location parameter $\mu$.

To fully leverage the flexibility of the parametrized diffeomorphism $\CP_{\text{total}}$,  hyperparameter $\theta = (\mu, R, \omega)$ should be tuned/optimised. A good choice of $\theta$ is one for which the pull-back of $\pi$ onto $B(R)$ is as close to the uniform measure as possible, yielding very fast mixing (see Section~\ref{subsec:non_asymptotic} below for the theoretical results). Inspired by~\cite{hoffman2019neutra,parno2018transport}, we approach this tuning problem by casting the choice of $\theta$ as a Variational Inference task of minimising the KL-divergence of the target with respect to an appropriate ``nicer'' density.  In~\cite{hoffman2019neutra,parno2018transport} it was shown that  VI  can significantly improve the performance of a sampling algorithm. In our case, the KL-divergence is minimised with respect to the uniform measure on $B(R)$, since this distributions offers rapid non-asymptotic mixing as discussed in Section~\ref{subsec:non_asymptotic} below. The VI procedure used in some numerical examples of Section~\ref{sec:sim} is described in Appendix~\ref{app:mobius_vi} below.

\subsubsection{M\"obius transformation} A simple, powerful and natural  automorphism $\mathcal{T}_\omega$ is the \textit{M\"obius transformation} on the unit ball $B(1)$ (also used in normalizing flows~\cite{rezende2020normalizing}). 
The parameter $\omega$ is a  vector $\omega = \delta \in B(1)$ (i.e., $|\delta| < 1$) and the transformation $\mathcal{M}_\delta : B(1) \to B(1)$ is given by  
\begin{equation}
\label{eq:mobius_unit}
y\mapsto \mathcal{M}_\delta(y) := \frac{(1 - |\delta|^2)y + (1 + 2\langle y, \delta \rangle + |y|^2)\delta}{1 + 2\langle y, \delta \rangle + |\delta|^2|y|^2},\qquad y\in B(1).
\end{equation}
Note that the M\"obius transformation $\mathcal{M}_\delta$ satisfies the assumptions on the diffeomorphism $\mathcal{T}_\omega$ required by our theory (see Theorem~\ref{thm:Uniform_ergodicity_unified} below).
In the examples of Section~\ref{sec:sim} below, the automorphism of the ball in~\eqref{eq:total_projection} of Section~\ref{subsec:automorphism} above in Algorithm~\ref{alg:DCS} uses the automorphism $\mathcal{M}_\delta$ in~\eqref{eq:mobius_unit}. Note that $\mathcal{M}_0$ is the identity on the unit ball $B(1)$. The Jacobian formula and the stochastic optimisation scheme for
this M\"obius family are given in
Appendix~\ref{app:mobius_vi} below.

\section{Convergence theory for the Diffeomorphic Contraction Sampler}
\label{sec:theoretical_results}

\subsection{Uniform Ergodicity}
\label{subsec:uniform_ergodicity}
In this section we show that Algorithm~\ref{alg:DCS}, based on the Markov kernels in Appendix~\ref{app:Algs},  is uniformly ergodic under the tail decay condition in~\eqref{eq:unif_ergodicity_pi_condition} below. In contrast to the non-asymptotic bounds in Theorem~\ref{thm:dcs_bounds} below,  requiring finer properties of the target (see~\ref{assump:convex}, \ref{assump:smooth}, \ref{assump:warm}, \ref{assump:isotropic},
\ref{assump:lipschitz}, \ref{assump:cold} in Section~\ref{subsec:non_asymptotic} below),
uniform ergodicity of Algorithm~\ref{alg:DCS} holds under a polynomial tail-decay Assumption~\eqref{eq:unif_ergodicity_pi_condition} only. 
The proof of our uniform ergodicity result, Theorem~\ref{thm:Uniform_ergodicity_unified}, is in Appendix~\ref{app:uniform_ergodicity_proofs} below.

\begin{thm}[Uniform Ergodicity of Algorithm~\ref{alg:DCS}]
\label{thm:Uniform_ergodicity_unified}
Let $\beta > 0$, $R > 0$, $\mu \in \R^d$. 
Assume that  strictly increasing, differentiable $\phi_\beta:[0,1)\to[0,\infty)$ satisfies~\eqref{eq:beta_condition_for_phi} and  $\CP_\beta:B(1)\to\R^d$ is a $C^1$-diffeomorphism.
Let $\mathcal{T}_\omega:\overline{B}(1)\to\overline{B}(1)$ be a $C^1$-diffeomorphism of the closed unit ball in $\R^d$ and consider the diffeomorphism $\RCtot:B(R)\to\R^d$ defined in~\eqref{eq:total_projection}. Let a continuous $\pi:\R^d\to(0,\infty)$  
satisfy the tail-decay condition:
\begin{equation}
\label{eq:unif_ergodicity_pi_condition}
\tag{A-$\beta$}
\limsup_{|x|\to\infty} \pi(x)|x|^{d+\beta} <\infty.
\end{equation}
\refstepcounter{equation}Consider 
the Markov chain $X=(X_n)_{n\in\N}$ with increments given by Algorithm~\ref{alg:DCS}
targeting a probability measure $\nu$ on $\R^d$ with density proportional to $\pi$. 
If the proposal kernel $Q$ in 
Algorithm~\ref{alg:DCS} is given by  any of the following algorithms,
\ref{alg:DCS-HnR}, \ref{alg:DCS-BW}, \ref{alg:DCS-RWM}  or \ref{alg:DCS-DW} with covariance-floored proposal
$\Sigma_\lambda$ in~\eqref{eq:regularised:Dikin_H} for some
$\lambda>0$,
then  there exist constants $C>0$ and $\rho\in(0,1)$ such that 
$$ \| \P_x(X_n \in \cdot) - \nu \|_{\mathrm{TV}}\leq C\rho^{n}\quad\text{for all $x\in\R^d$ and $n\in\N$,} $$
making the chain $X$ \textbf{uniformly ergodic} on $\R^d$.
\end{thm}

\begin{rem}
Theorem~\ref{thm:Uniform_ergodicity_unified} shows that uniform ergodicity of Algorithm~\ref{alg:DCS} is determined solely by the choice of the parameter $\beta$. Equivalently, whenever condition~\eqref{eq:unif_ergodicity_pi_condition} is satisfied, Algorithm~\ref{alg:DCS} is uniformly ergodic for every admissible choice of the remaining parameters in~\eqref{eq:total_projection}.
\end{rem}

\begin{rem}
\label{rem:dikin_dichotomy}
The parameter $\lambda$ in Algorithm~\ref{alg:DCS-DW} distinguishes the
classical Dikin walk studied in~\cite{kook2024gaussian} from the covariance-floored version used in the uniform
ergodicity theorem. When $\lambda=0$ and $H=\nabla^2\varphi$, where
$\varphi$ is the logarithmic barrier in~\eqref{eq:local_metric}, we recover
the classical Dikin walk, for which rapid warm-start non-asymptotic bounds are
available on the ball.\footnote{For any $C^2$-function $f$ on a domain in $\R^d$, the Hessian of $f$ is a matrix-valued function, denoted by $\nabla^2 f$, consisting of the second partial derivatives of $f$.} When $\lambda>0$,
the proposal covariance
       $ \Sigma_\lambda(z)
        =
        \left(H(z)^{-1}+\lambda I_d\right)\gamma^2/d$
is uniformly bounded from below on $B(R)$. This removes the singularity of the variance of
the classical Dikin covariance near the boundary and is the version used in
Theorem~\ref{thm:Uniform_ergodicity_unified}.
\end{rem}

\begin{rem}
\label{rem:uniform_langevin}
Under mild assumptions, Theorem~\ref{thm:Uniform_ergodicity_unified} establishes uniform ergodicity for the class of  samplers not relying on the gradient information of the target. In contrast, establishing uniform ergodicity for Algorithm~\ref{alg:DCS} with the gradient-based Markov kernel $Q$ in  Algorithm~\ref{alg:DCS-LMC} would require stronger regularity assumptions on $\nabla \log \pi_B$. Answering these important questions would entail a careful analysis of the interaction between the diffeomorphism $\RCtot$ and the gradient structure of the target density proportional to $\pi$. A detailed theoretical and  computational study of Algorithm~\ref{alg:DCS} with the Langevin-type Markov kernel in  Algorithm~\ref{alg:DCS-LMC} is left for the future research.
\end{rem}

\subsection{Non-asymptotic bounds for Algorithm~\ref{alg:DCS}}
\label{subsec:non_asymptotic}
This section utilises the properties of the transformed process along with the established theory~\cite{lovasz2004hit,MR2309621,jiang2024regularized,MR3802303,kook2024gaussian} 
on the total-variation mixing times of Markov chains given by the algorithms in Appendix~\ref{app:Algs} below, to deduce non-asymptotic convergence bounds for the corresponding instances of the Diffeomorphic Contraction Sampler~\ref{alg:DCS}.
The strategy is simple:
non-asymptotic convergence bounds for 
the Markov chain $(X_n)_{n\in\N}$ in $\R^d$  with increments given by Algorithm~\ref{alg:DCS}
can be obtained by studying the 
transformed chain $(\RCtot^{-1}(X_n))_{n\in\N}$ for the diffeomorphism $\RCtot^{-1}$ in~\eqref{eq:total_inverse} above.

The non-asymptotic theoretical results require more structure (e.g. log-concavity of the transformed target measure $\nu_B$) than their asymptotic counterparts in Theorem~\ref{thm:Uniform_ergodicity_unified} above, which hold under comparatively milder asymptotic conditions. We begin by stating the key assumptions on the fully preconditioned target $\pi_{B}$, which allow us to obtain the non-asymptotic convergence results on the projected sampling algorithms.

Recall that $\nu_B$ is the probability
measure on $B(R)$ (given by the pull-back of the target $\nu$ on $\R^d$) with density proportional to $\pi_B:B(R)\to(0,\infty)$. 
Let $U_B \coloneqq -\log \pi_{B}$ be the $C^2$-potential of $\pi_B$ and, for any $a\in(0,\infty)$, denote by $\mathcal{L}(a) \coloneqq \{z\in B(R): \pi_B(z)\ge a\}$ the level sets of $\pi_B$. Consider the following properties the measure $\nu_{B}$ might possess.

\begin{description}[leftmargin=2cm, style=standard]

    \item[\customlabel{assump:convex}{\textcolor{red}{(Conv)}}] \textbf{Convexity.} 
    The function $U_B$ is convex: 
    \[ U_B(y) \ge U_B(z) + \langle \nabla U_B(z), y - z \rangle\quad\text{for all $z,y\in B(R)$.} \]

    \item[\customlabel{assump:smooth}{\textcolor{red}{(Smooth)}}] \textbf{$M$-Smoothness.} 
    The gradient $\nabla U_B$ is $M$-Lipschitz continuous:
    \[ |\nabla U_B(z) - \nabla U_B(y)| \le  M|z - y|\quad\text{for all $z,y\in B(R)$.}  \]

    \item[\customlabel{assump:lipschitz}{\textcolor{red}{(Lip)}}] \textbf{$L$-Lipschitz Continuity.} 
    The function $U_B$ is $L$-Lipschitz continuous:
    \[ |U_B(z) - U_B(y)| \le L |z - y|\quad\text{for all $z,y\in B(R)$.} \]

\item[\customlabel{assump:isotropic}{\textcolor{red}{(Iso)}}]
\textbf{$C$-isotropy.}
The target $\nu_B$ is $C$-isotropic for some constant  $C\in(0,\infty)$ if
\[
        C^{-1}\leq 
        \int_{B(R)}\langle u, z\rangle^2 \nu_B(z) dz
        \leq C, \qquad \text{for every unit vector $u\in \mathbb {S}^{d-1}\subset\R^d$.}
\]
    \item[\customlabel{assump:warm}{\textcolor{red}{(Warm)}}] \textbf{Warm Start.} 
    The initial distribution $P_0$ is sufficiently close to $\nu_{B}$: there exists a constant $H_0>0$ satisfying
    \[ P_0(A) \leq H_0 \nu_{B}(A) \quad\text{for all $A\in \mathcal{B}(B(R))$.} \]
 \item[\customlabel{assump:cold}{\textcolor{red}{(Cold)}}] \textbf{Cold Start.} For constants $\underline{r},\overline{r}>0$ we have: $\int_{B(R)} |z'|^2 \nu_B(z')dz'\leq \overline r^2$ and if $\nu_B(\mathcal{L}(c))\geq 1/8$ for some $c>0$ then $\mathcal{L}(c)\supset  z'+ B(\underline r)$ for some $z'\in B(R)$.
    The starting point $z\in B(R)$ (i.e.,  $P_0 = \delta_z$) is at distance at least $m>0$  from boundary of the level set $\mathcal{L}(\pi_B(z)/2)$ and satisfies $\pi_B(z)\geq \varrho^d\sup_{z'\in B(R)}\pi_B(z')$ for some constant $\varrho\in(0,1]$. 
\end{description}

\begin{rem}
\label{rem:non_asympt_assumpions}
Since $U_B \in C^2$, Assumption~\ref{assump:convex} is equivalent to the Hessian being positive semi-definite: $\nabla^2 U_B(z) \succeq 0$.\footnote{By definition, for matrices $A,B\in\R^{d\times d}$ we have $A\succeq B$ if $A-B$ is a non-negative definite matrix.} Assumption~\ref{assump:smooth} implies an upper bound on the curvature: $\nabla^2 U_B(z) \preceq M \mathbf{I}$, while Assumption~\ref{assump:lipschitz} implies the norm of the gradient is bounded: $|\nabla U_B(z)| \le L$. Assumption~\ref{assump:isotropic} quantifies how far $\nu_{B}$ is  from an isotropic measure.  Assumption~\ref{assump:cold} guarantees that $\nu_B$ is sufficiently well ``spread out'' and that the initial point lies in the part of $B(R)$ with sufficient $\nu_B$-mass.  
\end{rem}

For any $\epsilon\in(0,1)$, define the
\textit{total-variation mixing time} of the~\ref{alg:DCS} chain $X$, started at an initial distribution $\varsigma$ on $\R^d$, by
\begin{equation}
\label{eq:mixing_def}
        \tau_X(\epsilon,\varsigma)
        :=
        \inf\left\{
        n\ge0:
        \left\| \P_{\varsigma}(X_n\in \cdot)-\nu\right\|_{\TV}\le\epsilon
        \right\}.
\end{equation}
The results in~\cite{MR2309621,lovasz2004hit,MR3802303,kook2024gaussian} and Corollary~\ref{cor:convergence_equivalence} (see Appendix~\ref{app:proofs_non_asymptotic} below) imply the following non-asymptotic bounds for Algorithm~\ref{alg:DCS}. 
The proof of Theorem~\ref{thm:dcs_bounds} is also in Appendix~\ref{app:mixing_proofs} below.

\begin{thm}[Non-Asymptotic Bounds for~\ref{alg:DCS}]
\label{thm:dcs_bounds}
Recall $\nu$ is  target distribution on $\mathbb{R}^d$  with density proportional to $\pi:\R^d\to(0,\infty)$, $X=(X_n)_{n \in \mathbb{N}}$ a Markov chain generated by Algorithm~\ref{alg:DCS} (with
initial distribution $\varsigma$ on $\R^d$)
targeting $\pi$ and $\pi_B:B(R)\to(0,\infty)$ is given in in~\eqref{eq:total_target}. 
Then for any $\epsilon\in(0,1)$ the following non-asymptotic bounds hold. 

\begin{enumerate}
    \item[(a)] \textbf{\ref{alg:DCS-BW}} Assume that $\pi_B$ satisfies Assumptions~\ref{assump:convex}, \ref{assump:isotropic}, \ref{assump:warm}  hold,  and that the Ball-Walk step size satisfies $ 0<\gamma
        \le \epsilon^2/(2^{10}eH_0^2\sqrt{Cd})$.  Then the mixing time of the chain $X$ on $\R^d$, started at $\varsigma=P_0(\RCtot^{-1}(\cdot))$,  satisfies
   $$ \tau_X(\epsilon,\varsigma)\leq
        10^{10}
        Cd^2\gamma^{-2}
        \log\frac{2H_0}{e\sqrt{C}\epsilon}.
    $$
    \item[(b)] \textbf{\ref{alg:DCS-HnR} (Warm start)} If $\pi_B$ satisfies Assumptions \ref{assump:convex}, \ref{assump:isotropic}, and~\ref{assump:warm}, then the mixing time of the chain $X$ in $\R^d$, started at $\varsigma=P_0(\RCtot^{-1}(\cdot))$, satisfies 
    $$
    \tau_X(\epsilon,\varsigma)\leq
        10^{30}C^4 H_0^4 d^3(\log (2H_0/\eps))^3\eps^{-4}.
    $$
 \item[(c)] \textbf{\ref{alg:DCS-HnR} (Cold start)} Fix $x\in\R^d$ and 
let $\pi_B$ satisfy Assumptions~\ref{assump:convex} and
\ref{assump:cold} (with some $z\in B(R)$). 
Then, the mixing time of $X$, started at $\varsigma=\delta_{x}$ with $x=\RCtot(z)$, satisfies
\[
\tau_X\left(\epsilon,
        \varsigma
        \right)
        \le
        10^{31}d^3\frac{\overline r^2}{\underline{r}^2}\left(\log\left(\frac{d\overline r^2}{\eps \underline r m\varrho}\right)\right)^5.
\]
    \item[(d)] \textbf{\ref{alg:DCS-LMC}} If $\pi_B$ satisfies Assumptions \ref{assump:convex}, \ref{assump:smooth}, and \ref{assump:lipschitz} and define $f := \max\{d,R\cdot M,R\cdot L\}/\eps$. Then, there exist constants $C_\textrm{LMC},c_{\textrm{LMC}},C_{\mathrm{step}},c_{\mathrm{step}}>0$, independent of $d,L,R,M,\eps$, such that if the step size in Algorithm~\ref{alg:DCS-LMC} equals $\eta = C_{\mathrm{step}}(\log f)^{c_{\mathrm{step}}} R^2/f^{12}$, and we have   
  $$ N\geq  
C_{\mathrm{LMC}}R^6f^{12}(\log f)^{c_{\mathrm{LMC}}},$$
the chain started from $X_0=x=\RCtot(0)$ satisfies
$\|\mathcal L_x(X_N)-\nu_B
\|_{\mathrm{TV}}
\leq\epsilon.$

    \item[(e)] \textbf{\ref{alg:DCS-DW}} Assume that the Markov kernel in Algorithm~\ref{alg:DCS-DW} uses
$\lambda=0$ and local metric $H=d \nabla^2\varphi$, where
$\varphi$ is defined in~\eqref{eq:local_metric}. Let $\pi_B$ satisfy Assumptions~\ref{assump:convex}, \ref{assump:smooth} and \ref{assump:warm} and the step size equal $\gamma= \min\{1,2d/(MR^2)\}^{1/2}$.
Then there exists as constant $C_\textrm{DW}>0$, independent of $d,R,M,H_0,\epsilon$, such that the mixing time of $X$ started at $\varsigma=P_0(\RCtot^{-1}(\cdot))$ satisfies
        $$\tau_X(\epsilon,\varsigma)
        \le
        C_{\mathrm{DW}}\,
        d^2\max\{
        1,MR^2/(2d)\}
        \log\left(H_0/\epsilon\right).$$
\end{enumerate}
\end{thm}

\begin{rem}
\label{rem:non_asymp_constant_dependence}
The assumptions of Theorem~\ref{thm:dcs_bounds} are formulated in terms of the transformed target $\pi_B$ on $B(R)$. Thus, the constants appearing in the non-asymptotic bounds depend on the parameters of the transformation in~\eqref{eq:total_projection}, namely  $\mu\in\R^d$, $R>0$, $\beta>0$ and the automorphism $\mathcal{T}_\omega$. 
Typically parameters $\beta,\mu$ and $\mathcal{T}_\omega$ do not exhibit a significant dimension dependence. In contrast, the radius of the ball $B(R)$ grows with dimension. 

The parameter $R$ appears in the bounds of Theorem~\ref{thm:dcs_bounds} both explicitly (Algorithms~\ref{alg:DCS-LMC} and~\ref{alg:DCS-DW}) and implicitly (through the constants in the assumptions on $\pi_B$). In a specific example it is possible to track this dependence on dimension through $R$.
For instance, setting $R^2=v+d$ in Example~\ref{ex:student_mixing} below yields a $C$-isotropic transformed target with $C$ proportional to a constant and hence a mixing-time bound proportional to $d^3$ for the warm start Algorithm~\ref{alg:DCS-HnR}. In contrast, choosing $R^2=v$, so that $R$ remains fixed as the dimension grows, yields a $C$-isotropic transformed target $C$ proportional to the dimension $d$, yielding a mixing-time bound of order $d^7$.
\end{rem}

\begin{rem}
The development of non-asymptotic bounds for the convergence of sampling algorithms on convex domains was originally driven by the problem of estimating the volume of a convex body~\cite{MR1608200,lovasz2004hit}. This topic continues to attract significant attention~\cite{lee2024eldan,MR5009943}. Most results are expressed in terms of the dimension $d$ and the ratio $R_K/r_K$ of the radii  of the circumscribed and inscribed balls ($R_K$ and $r_K$, respectively) of the convex body $K\subset \R^d$. For a Euclidean ball, the ratio $R_K/r_K=1$ is clearly minimised  over all convex bodies.

For general convex bodies, the quotient $R_K/r_K$ is typically proportional to a polynomial of the dimension $d$.
Using stochastic localization,   recent work~\cite{MR5009943} shows that the   mixing bound for the uniform distribution on a convex body
depends on $R_K/r_K$ through the factor $\log(R_K/r_K)$  (contributing only a logarithmic factor in dimension $d$). However, for non-uniform  densities on a convex body, the general best known upper bound on $R_K/r_K$ is quadratic~\cite[Thm~1.1]{lovasz2004hit} in dimension $d$. Even though Algorithm~\ref{alg:DCS} could be deployed on a convex body different from a Euclidean ball, these results suggest that a transformation to a Euclidean ball  is the most suitable for achieving rapid mixing.
\end{rem}

\begin{rem}
\label{rem:dikin_non_asymp}
The assumption $\lambda=0$ in the case  Algorithm~\ref{alg:DCS-DW} (based on  a Dikin Walk) in Theorem~\ref{thm:dcs_bounds} appears not to be necessary as suggested by analogous non-asymptotic bounds for regularised Dikin Walks given in~\cite{jiang2024regularized}. The result in~\cite[Thm~3.1]{jiang2024regularized} is not directly applicable in our setting as it uses a slightly different regularisation than the one in Algorithm~\ref{alg:DCS-DW} and  covers distributions truncated on polytopes with finitely many faces. 
\end{rem}

\begin{rem}
\label{rem:langevin_comment_eldan}
Despite having worse theoretical mixing bounds than the other algorithms, Langevin Monte Carlo algorithm~\ref{alg:DCS-LMC} empirically exhibits mixing times comparable to Hit-and-Run with cold start~\cite{MR3802303} (see also~\cite{brosse2017sampling} for related results on Langevin-based sampling).
While strict non-asymptotic bounds are not currently available for the standard RWM with Gaussian proposals on a convex set (cf.,~\ref{alg:DCS-RWM}), the recent paper~\cite{kook2024and}, introducing a modification of the Gaussian proposals,  proves that its non-asymptotic bounds match those of the Ball Walk (cf.,~\ref{alg:DCS-BW}).
\end{rem}

\subsubsection{Which targets $\pi$ satisfy non-asymptotic assumptions after transformation~\eqref{eq:total_projection}?}
\label{subsec:geometric_properties}
It is natural to enquire which transformed target densities  
$\pi_B$ in~\eqref{eq:total_target} satisfy the geometric assumptions required for the non-asymptotic mixing rates in Theorem~\ref{thm:dcs_bounds} above. With this in mind,  we set
$U\coloneqq -\log \pi$ on $\R^d$
and analyse how the geometry of the potential $U$ is affected by  transformation~\eqref{eq:total_projection} to the bounded domain $B(R)$. 
In order to keep the discussion focused, in this section we work with $\CP_\beta$ given in~\eqref{eq:compact_projection_first}.
Using identities~\eqref{eq:total_jacobian},~\eqref{eq:total_target} and
\[
1-|w|^\beta=(1+|v|^\beta)^{-1},\quad\text{where }v=(\RCtot(z)-\mu)/R,
\quad
w=\mathcal T_\omega(z/R)\quad\text{for $z\in B(R)$,}
\]
the transformed potential $U_B= -\log \pi_B$ on $B(R)$ admits the  representation $U_B(z)=\tilde U_B(z)-\log J_{\mathcal T_\omega}(z/R)$  with
\begin{equation}
\label{eq:transformed_potential}
\tilde U_B(z)
\coloneqq
U(\RCtot(z))
-
\left(\frac{d+\beta}{\beta}\right)
\log\left(1+\left|\frac{\RCtot(z)-\mu}{R}\right|^\beta\right).
\end{equation}
Crucially, \(\mathcal T_\omega\) is a \(C^1\)-diffeomorphism with bounded Jacobian on the closed unit ball \(\overline{B}(1)\), the asymptotic behaviour  of $U_B$ as $|z|\uparrow R$ is governed by the corresponding asymptotics of \(\tilde U_B\) in~\eqref{eq:transformed_potential}. 

Consider Assumption~\ref{assump:convex} for $\pi_B$.
As  $\CP_{\text{total}}(z)= \mu + R \cdot \CP_\beta(\mathcal{T}_\omega(z/R))\to\infty$ when $|z|\uparrow R$, the relation in~\eqref{eq:transformed_potential} implies  that the asymptotic convexity of $U_B=-\log \pi_B$ requires the inequality
\begin{equation}
\label{eq:beta_decay_condition}
U(x) = -\log \pi(x)\geq (d+\beta)\log |x| \quad\text{for all large $|x|$.}
\end{equation}
Thus the necessary condition for convexity  near the boundary of $B(R)$ for any target $\pi$ is to decay faster than $|x|^{-(d+\beta)}$ as $|x|\to\infty$. In particular, all targets $\pi$, such that\footnote{By definition, for positive functions $f,g$ we have $f\asymp g$ on a neighbourhood of  $x_0$ if there exist constants $c,C>0$ such that $c\leq f(x)/g(x)\leq C$   for all $x$ sufficiently close to $x_0$.}  $\pi(x) \asymp |x|^{-(d+p)}$ with $p\geq \beta$ as $|x|\to\infty$ will be transformed to the density $\pi_B$ satisfying~\ref{assump:convex} as $|z|\uparrow R$. We stress that the condition in~\eqref{eq:beta_decay_condition}  coincides (up to an additive constant) with the sufficient condition in~\eqref{eq:unif_ergodicity_pi_condition} for the uniform ergodicity of Algorithm~\ref{alg:DCS} in Theorem~\ref{thm:Uniform_ergodicity_unified} above.

 Global convexity of $\pi_B$ on $B(R)$ in~\eqref{eq:student_def} induced by the diffeomorphism $\RCtot$    with $\mathcal{T}_\omega=\Id$ holds for well-behaved unimodal target distributions, such as symmetric student-$t$ distributions (for the definition of the Student-$t$ see~\eqref{eq:student_def} below). However, by employing a suitable transformation $\RCtot$ with a non-trivial automorphism $\mathcal{T}_\omega:B(1)\to B(1)$, complex multimodal distributions can be mapped onto log-concave targets (see Appendix~\ref{app:diffeo_convex} below for general conditions on targets and diffeomorphisms that exhibit this property).

Satisfying the remaining assumptions is less critical, since they do not appear in the non-asymptotic bounds for all samplers. Assumptions~\ref{assump:smooth} and~\ref{assump:lipschitz} are the most restrictive: they require the target density $\pi(x)$ to have heavy polynomial tails that \emph{precisely match} the asymptotic volume expansion of the mapping up to second-order accuracy. In particular,
\begin{equation}
\label{eq:asym_lip}
\pi(x)\,|x|^{d+\beta} \to a_0 \in (0,\infty) 
\qquad \text{ as } |x|\to\infty .
\end{equation}
If~\eqref{eq:asym_lip} does not hold, then $\tilde U_B(z)\to\infty$ or $\tilde U_B(z)\to -\infty$ as $|z|\to R$, violating the $M$-smoothness and $L$-Lipschitz conditions. On the ball $B(R)$, failure of~\eqref{eq:asym_lip} implies that either $\pi_B(z)\to 0$ or $\pi_B(z)\to\infty$ as $|z|\to R$. By contrast, Assumption~\ref{assump:isotropic} imposes no specific tail requirement on $\pi_B$. It is satisfied with small $C$ for isotropic or nearly isotropic targets. For highly skewed targets (e.g., the Rosenbrock ``banana'' density~\cite{MR1828504,rosenbrock1960automatic}), $\mathcal{T}_\omega$ may be used to balance its mass within the unit ball, thereby increasing the lower bound $1/C$.

\begin{rem}{Summary of the geometric trade-offs.}
\label{rem:trade_off}
The diffeomorphic framework establishes a mathematical trichotomy based on the original target measure.
\begin{itemize}
\item \textbf{Tails of $\pi$ with decay faster than $|x|^{-(d+\beta)}$:} this choice of $\beta>0$ yields $\pi_B$ satisfying~\ref{assump:convex} and~\ref{assump:isotropic}. The transformation $\RCtot$ maps these lighter tails into an infinitely steep potential $U_B$ at the boundary.  This property precludes global bounds in Assumptions~~\ref{assump:lipschitz} and~\ref{assump:smooth}.

\item \textbf{Tails of $\pi$ with decay  $|x|^{-(d+\beta)}$:} if $\pi$ asymptotically matches the mapping's expansion in the tails, the transformation $\RCtot$ makes the potential $U_B$ bounded on $B(R)$. This perfectly balances the geometry in the tails, ensuring global bounds required in Assumptions~\ref{assump:smooth} and~\ref{assump:lipschitz}.

\item \textbf{Tails of $\pi$ with decay slower than $|x|^{-(d+\beta)}$:} if the transformation parameter $\beta$ is chosen such that the target tails are heavier than the mapping's expansion, Assumption~\ref{assump:convex} is violated (see~\eqref{eq:transformed_potential} and the discussion that follows it). In this regime, uniform ergodicity fails, resulting in poor mixing performance.
\end{itemize} 

In summary, one should ideally choose $\beta$ to perfectly match the decay rate of $x \mapsto \pi(x)|x|^d$ as $|x|\to\infty$. However, since an exact match is rarely feasible,  it is always preferable to choose $\beta$ that is smaller, rather than larger, compared to the decay rate of  $\pi$. This conservative choice ensures uniform ergodicity and the asymptotic convexity necessary for fast mixing. Conversely, choosing $\beta$ too large  results in the lack of asymptotic convexity and the failure of uniform ergodicity (cf., Theorem~\ref{thm:Uniform_ergodicity_unified} below).
\end{rem}

\subsection{Related literature in the context of our contributions}
\label{subsec:literature}

Sampling from heavy-tailed distributions is notoriously difficult due to the compounding effects of unbounded state spaces, high dimensionality, and vanishing gradients in the tails. As a result, standard MCMC samplers perform poorly in this regime, a limitation both empirically observed~\cite{Roberts07,MR1796485} and theoretically demonstrated~\cite{livingstone2019geometric,roberts1996geometric,brevsar2025central}.

The first step towards a solution to the heavy-tailed sampling problem in~\cite{MR3097969} proposes an automorphism of $\R^d$ that transforms a heavy-tailed $\pi$ to a target with super-exponential tail decay and then runs a RWM algorithm on the transformed distribution. While~\cite{MR3097969} yields a geometrically ergodic algorithm, it fails to achieve uniform ergodicity. Moreover, the framework in~\cite{MR3097969} cannot accommodate substituting RWM with algorithms with better scaling in dimension: if the state space transformation makes the tails of $\pi$ lighter than Gaussian, the resulting gradients would cause the unadjusted Langevin algorithm to diverge~\cite{roberts1996geometric} and MALA/HMC not to be geometrically ergodic. The idea in~\cite{MR3097969}  was employed in~\cite{MR3911112} to obtain exponential convergence rate of the Bouncy Particle Sampler to heavy-tailed targets. However, despite its strong theoretical properties, to the best of our knowledge the algorithm in~\cite{MR3911112} lacks an efficient implementation.

The next important step~\cite{Yang24} in heavy-tailed sampling transforms the density $\pi$ on $\mathbb{R}^d$ to a density on the punctured unit sphere $\mathbb{S}^d\setminus \{N\}$ in $\R^{d+1}$ (where $N$ denotes the North Pole)  via the stereographic projection and runs a RWM on $\mathbb{S}^d$. When $\pi$  possesses $d$ finite moments, this approach yields a uniformly ergodic algorithm, improving upon the geometric ergodicity in~\cite{MR3097969}. However, an important feature of the stereographic projection map is that when the target $\pi$ has fewer than $d$ finite moments (e.g. $\liminf_{|x|\to\infty}\pi(x)|x|^{d+d}=\infty$), the transformed density is unbounded around the North Pole  $N\in\mathbb{S}^d$. Consequently, if the target distribution has fewer than $d-1$ moments, the algorithm in~\cite{Yang24} exhibits slower polynomial convergence than the standard RWM~\cite[p.~6, Table~1]{brevsar2025central}. This pathology was partially remedied in~\cite{grazzi2026sub} through a generalized stereographic projection that maps the Euclidean space to a spherical cap of a hypersphere in $\R^{d+1}$ (rather than the entire punctured sphere $\mathbb{S}^d\setminus \{N\}$). This modification yields a uniformly ergodic sampler  for any target with at least one finite moment. The cost of this adjustment, however, is that the sampler must navigate a curved manifold with a boundary, resulting in a higher numerical cost per sample (see Section~\ref{subsec:comp_stereo} below for simulation tests and comparison with Algorithm~\ref{alg:DCS}).

Our method builds on the ideas in~\cite{brevsar2025central,Yang24,MR3097969} by introducing a novel transformation of the state space, inspired by the theory of sampling algorithms on convex bodies~\cite{jiang2024regularized,MR2309621,lovasz2004hit}. Specifically, we map the unbounded Euclidean space to a bounded ball $B(R)$ in $\R^d$ via a  diffeomorphism. This mapping allows us to leverage powerful algorithms, notably the Ball Walk and Hit-and-Run~\cite{MR2309621,lovasz2004hit} and others, which are known for their fast mixing and numerical efficiency on a round ball $B(R)$~\cite{MR5009943,MR2309621}. Crucially, our map $\RCtot$ in~\eqref{eq:total_projection} is composed of a radial contraction and an automorphism of the ball $B(R)$. The radial contraction yields uniform ergodicity for target distributions with arbitrarily heavy polynomial tails,  bypassing the moment conditions of existing stereographic approaches. The automorphism of $B(R)$ functions as a  flexible preconditioning tool, allowing the map $\RCtot$ to flatten highly skewed and multimodal targets while preserving the uniform ergodicity of Algorithm~\ref{alg:DCS}.
Summarising, Algorithm~\ref{alg:DCS} improves the state-of-the-art in heavy-tailed sampling in the following four ways.
\smallskip

\noindent \textbf{Breaking the heavy-tail barrier.} Algorithm~\ref{alg:DCS} is uniformly ergodic  for targets with arbitrarily heavy polynomial tails.

\smallskip

\noindent \textbf{High numerical efficiency and fast mixing.} Maintains a low computational overhead per sample with fast non-asymptotic mixing guarantees due to the simple geometry of the ball $B(R)$.

\smallskip

\noindent \textbf{Structural flexibility.} Easily accommodates a wide variety of Markov kernels, including RWM, gradient-based Markov Chains, and slice sampling via Hit-and-Run.

\smallskip

\noindent \textbf{Outperforms state-of-the-art.} Numerical examples in Section~\ref{sec:sim} indicate that Algorithm~\ref{alg:DCS} outperforms the No-U-Turn Sampler on real-world target laws with polynomial tails and complex funnel geometry as well as spherical projection samplers~\cite{Yang24,grazzi2026sub} on standard benchmarks for heavy-tailed sampling.

\section{Numerical performance of Algorithm~\ref{alg:DCS}}
\label{sec:sim}

In this section we test the various instances of Algorithm~\ref{alg:DCS} against the state-of-the-art samplers. In Section~\ref{subsec:art_data} and~\ref{subsec:data_real} we compare the~\ref{alg:DCS-BW} and~\ref{alg:DCS-HnR} against the No-U-Turn-Sampler, respectively. In Section~\ref{subsec:art_data} the target is a skewed student-t distribution~\cite{azzalini2003distributions}. Section~\ref{subsec:data_real} compares the samplers on real-world posterior distributions from PosteriorDB database~\cite{posteriordb} exhibiting both heavy tails and funnel phenomena. In Section~\ref{subsec:comp_stereo} we compare~\ref{alg:DCS-BW} and~\ref{alg:DCS-RWM} against the recent Stereographic sampler~\cite{grazzi2026sub} by replicating the experiments from~\cite[Sec~4]{grazzi2026sub}.
Throughout Section~\ref{sec:sim} we work with $\CP_\beta$ given by the formula in~\eqref{eq:compact_projection_first} above.

\subsection{Skewed student-$t$ targets: comparison of~\ref{alg:DCS-BW} and No-U-Turn Sampler}
\label{subsec:art_data}

As a first test, we investigate whether the Ball-Walk Algorithm~\ref{alg:DCS-BW} effectively mitigates the geometric limitations that hinder traditional MCMC methods in heavy-tailed settings~\cite{brevsar2025central}. We benchmark Algorithm~\ref{alg:DCS-BW} against the No-U-Turn Sampler (NUTS)~\cite{hoffman2014no} on the task of targeting a high-dimensional skewed Student-$t$ distribution. Following~\cite{azzalini2003distributions}, consider the symmetric Student-$t$ distribution skewed by the Gaussian distribution function,
\begin{equation}
\label{eq:skewed_student_t}
 \pi(x) \propto \left( 1 + \frac{|x|^2}{v} \right)^{-\frac{v+d}{2}} \Phi\left( \langle\alpha, x\rangle \sqrt{\frac{v+d}{v + |x|^2}} \right),\quad\text{where $\Phi(\cdot)\coloneqq\int_{-\infty}^\cdot\exp(-u^2/2)\ud u/\sqrt{2\pi}$,}
\end{equation}
$x \in \mathbb{R}^d$,  $d=200$, $v = 3$ degrees of freedom and $\alpha\in \mathbb{R}^d$ dictates the skewness direction, with first and second coordinates $20$ and $-30$, respectively, and all remaining coordinates equal to zero.

By exhibiting both heavy tails and multidimensional asymmetry, the target distribution in~\eqref{eq:skewed_student_t} by definition violates the symmetric and light-tailed geometric assumptions under which many sampling algorithms are known to exhibit fast mixing (see e.g,~\cite{chewi2025log}). Both Algorithm~\ref{alg:DCS-BW} and NUTS~\cite{hoffman2014no} first undergo a parameter-tuning stage. Algorithm~\ref{alg:DCS-BW} is optimised by algorithm  in Appendix~\ref{app:opt_alg} below, while NUTS utilises JAX implementation for parameter tuning. Once tuned, both algorithms are run with the same time budget on the same CPU hardware.

\begin{figure}[htbp]
    \centering
    \includegraphics[width=\textwidth]{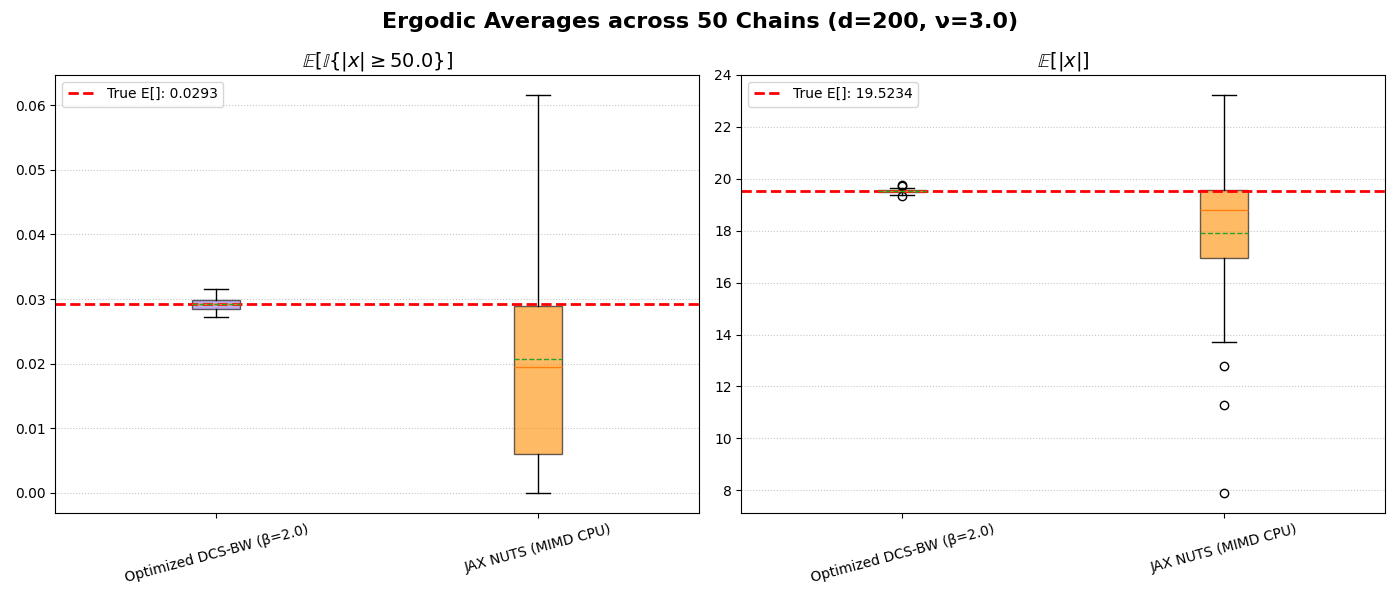}
    \caption{Boxplots of ergodic averages across 50 independent chains for a 200-dimensional skewed Student-$t$ target with $v=3$. The left (resp. right) panel shows the estimated tail probability $\mathbb{E}[\mathbb{I}\{|X| \geq 50.0\}]$ (resp. expected norm $\mathbb{E}[|X|]$). Red dashed lines indicate the exact analytical ground-truth values of 0.0293 and 19.5234, respectively. The results compare the optimised Algorithm~\ref{alg:DCS-BW} with $\beta=2$ (in~\eqref{eq:total_projection} and~\eqref{eq:total_inverse} of Algorithm~\ref{alg:DCS}) and the JAX implementation of NUTS~\cite{hoffman2014no, phan2019composable, jax2018github} sampler, with a 20-second CPU budget for each. Boxplot elements denote the median (solid orange line), mean (dashed green line), interquartile range (box), non-outlier range (whiskers), and outliers (white circles).}
    \label{fig:performance_NUTST}
\end{figure}

In this test, Algorithm~\ref{alg:DCS-BW} achieved substantially lower mean squared error (MSE) than NUTS, see Figure~\ref{fig:performance_NUTST}. These results suggest that Algorithm~\ref{alg:DCS-BW} efficiently samples from a high-dimensional skewed heavy-tailed distributions, succeeding in a setting where a widely used state-of-the-art sampler exhibits structural failures. 

To provide further context for the numerical results in Figure~\ref{fig:performance_NUTST},  we conduct an additional experiment targeting distribution~\eqref{eq:skewed_student_t} with the same parameters as in Figure~\ref{fig:performance_NUTST}. This time, we include Random Walk Metropolis (RWM) samplers utilizing both Gaussian and symmetric $\alpha$-stable ($\alpha=0.5$) proposal distributions.  For heavy-tailed target distributions, the RWM with infinite-variance proposal has been shown to perform better than the  RWM with Gaussian proposal~\cite{brevsar2025central}. The results demonstrate that while the highly optimized JAX implementation of NUTS outperforms both RWM methods, Algorithm~\ref{alg:DCS-BW} with a naive implementation considerably outperforms all of them. Figure~\ref{fig:performance_RWM} highlights the strength of NUTS~\cite{hoffman2014no} as the strongest standard baseline in this context, 
significantly outperforming RWM. This demonstrates that the computation of gradients in the tails is not necessarily harmful for the performance of the algorithm. Figure~\ref{fig:performance_RWM} also demonstrates  the significant potential of Algorithm~\ref{alg:DCS} given its suboptimal implementation (in contrast to the JAX implementation of NUTS~\cite{hoffman2014no, phan2019composable, jax2018github}).

\begin{figure}[htbp]
    \centering

    \includegraphics[width=\textwidth]{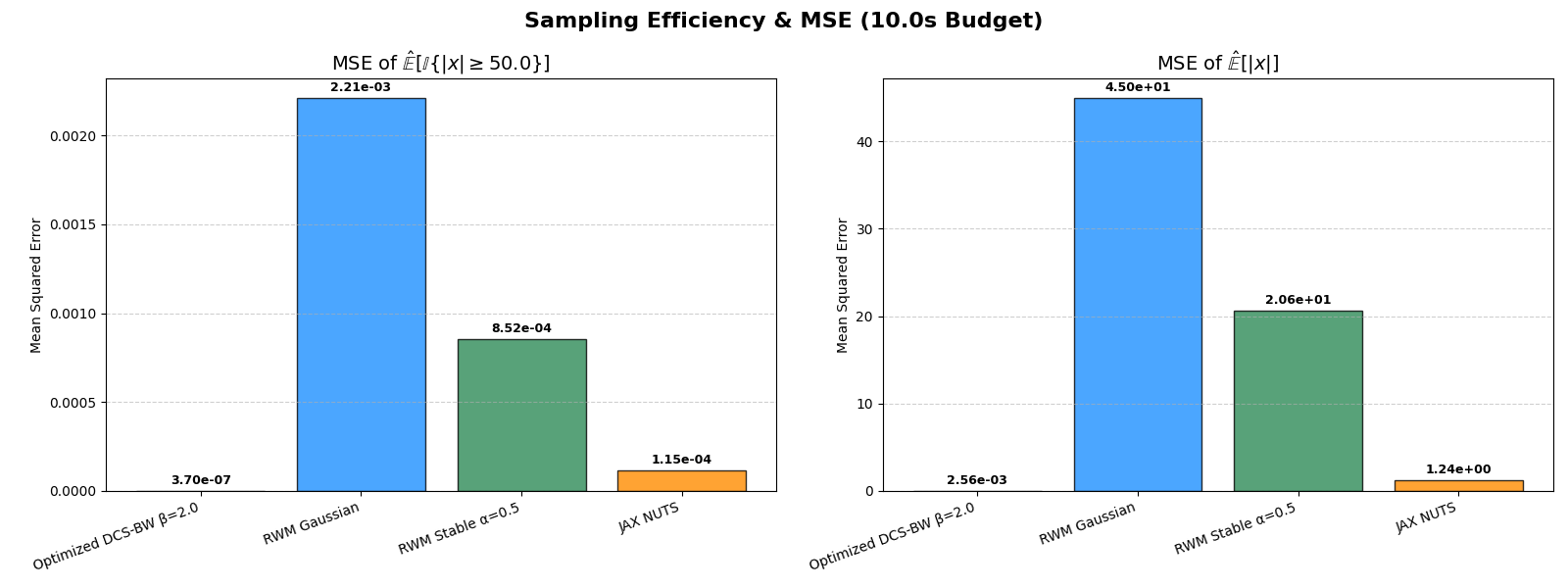}
    \caption{Comparison of performance and mean squared error across 50 independent chains (10 seconds of CPU time each) for the sampling from the heavy-tailed target~\eqref{eq:skewed_student_t} (a 200-dimensional skewed Student-$t$ target with $v=3$). NUTS clearly outperforms both RWM 
    algorithms suggesting that, in spite of its dependence on the gradient (which in this case vanishes in the tails), it remains the state-of-the-art among base-line algorithms in STAN. However, a naive implementation of Algorithm~\ref{alg:DCS-BW} outperforms JAX implementation of NUTS target in~\eqref{eq:skewed_student_t}.}
    \label{fig:performance_RWM}
\end{figure}

\subsection{Comparison of Algorithm~\ref{alg:DCS-HnR} and NUTS on real-world data: PosteriorDB}
\label{subsec:data_real}

In order to evaluate Algorithm~\ref{alg:DCS} on difficult real-world posterior distributions, we consider two benchmark models from PosteriorDB~\cite{magnusson2024posteriordb,posteriordb} with heavy-tailed features. The datasets, model specifications, and posterior reference values used as ground truth are all taken from PosteriorDB. Both
posterior distributions considered in this section contain
parameters with constrained support and therefore formally fall outside the
full-space setting assumed by our theoretical results. The numerical results presented in this subsection demonstrate the robustness of Algorithm~\ref{alg:DCS} beyond the regime
covered by our theory. We now describe the two models.

\begin{enumerate}
    \item[(I)]
    The Centred Eight Schools model studies the effects of SAT programs across eight schools. In the centred parametrisation, we have
    \[
        y_i \mid \theta_i \sim \mathcal{N}(\theta_i,\sigma_i^2),
        \qquad
        \theta_i \mid \mu,\tau \sim \mathcal{N}(\mu,\tau^2),
        \qquad
        \mu \sim \mathcal{N}(0,5^2),
        \qquad
        \tau \sim \mathrm{Cauchy}^+(0,5),
    \]
    where $\mathrm{Cauchy}^+(0,5)$  is a half-Cauchy prior  with density proportional to $u\mapsto 1/(25+u^2)$ for $u\in\R_+$ (i.e., a modulus of Cauchy distribution  centred at $0$ with scale parameter equal to $5$) and $\mathcal{N}(a,b)$ is a Gaussian distribution in $\R$ with mean $a\in\R$ and variance $b\in(0,\infty)$. 
    The parameters $\sigma_i$ are standard deviations of $y_i$, which  are fixed and given in the dataset, see PosteriorDB~\cite{magnusson2024posteriordb}.
    The posterior distribution, for parameters $\mu\in \R$, $\tau\in \R_+$ and $\theta=(\theta_1,\dots,\theta_8)\in\R^8$, therefore satisfies
    \[
        \pi(\mu,\tau,\theta \mid y)
        \propto
        \left[
        \prod_{i=1}^8
        \mathcal{N}(y_i \mid \theta_i,\sigma_i^2)
        \mathcal{N}(\theta_i \mid \mu,\tau^2)
        \right]
        \mathcal{N}(\mu \mid 0,5^2)
        \mathrm{Cauchy}^+(\tau \mid 0,5).
    \]
    Sampling from the posterior is difficult for two  reasons. First, the half-Cauchy prior is heavy-tailed, with density  decaying asymptotically as $\tau^{-2}$ as $\tau \to \infty$,
    leading to a heavy-tailed posterior. Second, the centred Gaussian hierarchy contributes the factor
    \[
        \prod_{i=1}^8 \mathcal{N}(\theta_i \mid \mu,\tau^2)
        \propto
        \tau^{-8}
        \exp\left\{
        -\frac{1}{2\tau^2}
        \sum_{i=1}^8(\theta_i-\mu)^2
        \right\},
    \]
    which forces $\theta_i \approx \mu$ for $\tau$ small (recall that $\tau\sim \mathrm{Cauchy}^+(0,5)$ takes small values often since the mode of the Cauchy density equals zero). Thus, the posterior combines heavy-tailed behaviour in with  funnel geometry~\cite{neal2003slice}: heavy-tailed for $\tau$ large and a singular behaviour for $\tau$ near $0$.

    \item[(II)]
The GARCH(1,1) model describes time-varying volatility in financial returns $r_t$:
\[
    r_t\mid h_t\sim\mathcal N(0,h_t),
    \qquad
    h_t=\omega+\alpha r_{t-1}^2+\beta h_{t-1},
    \qquad
    \omega>0,\quad \alpha,\beta\geq0,\quad \alpha+\beta<1 .
\]
The PosteriorDB model uses a flat prior on the constrained parameter space:
\[
    \pi(\omega,\alpha,\beta\mid r_{1:T})
    \propto
    \prod_{t=1}^T
    h_t^{-1/2}
    \exp\left\{-\frac{r_t^2}{2h_t}\right\}
    \mathbb I\{\omega>0,\alpha,\beta\geq0,\alpha+\beta<1\}.
\]
The main difficulty is near the stationarity boundary. Writing
\[
    H=\frac{\omega}{1-\alpha-\beta},
\]
the unconditional variance scale $H$ diverges as $\alpha+\beta\uparrow1$. Along this large-volatility ridge we have  $h_t=O(H)$ as $\alpha+\beta\uparrow1$. Thus the posterior satisfies
\[
    \pi(\omega,\alpha,\beta)\asymp H^{-T/2},
\]
making the tail polynomial of order $T/2$. The challenge is therefore a long heavy-tailed ridge for $(\omega, \alpha, \beta)$ with $\alpha+\beta$ close to $1$, together with the hard constraint $\alpha+\beta<1$.
\end{enumerate}

To test the robustness of our approach, we benchmarked Algorithm~\ref{alg:DCS-HnR} (without the parameter tuning optimisation step in Section~\ref{sec:optimization} above) against the  JAX~\cite{hoffman2014no, phan2019composable, jax2018github} implementation of NUTS. In these experiments, Step~4 of
Algorithm~\ref{alg:DCS-HnR} was implemented using a valid
univariate slice-sampling transition on the feasible chord rather
than an exact draw from the one-dimensional conditional
distribution. For each model, we ran NUTS for 150,000 post-warmup sampling steps and recorded the total execution time. We then allocated the exact same time budget to Algorithm~\ref{alg:DCS-HnR}. The accuracy of the algorithms was measured by calculating the Root Mean Square Error (RMSE) of the estimators of the first and second moments of the posterior across 50 independent chains  (ground-truth values are given in PosteriorDB~\cite{magnusson2024posteriordb}).

\begin{table}[ht]
\centering
\begin{tabular}{l|cc|cc}
\toprule
\multirow{2}{*}{\textbf{Model}} & \multicolumn{2}{c|}{\textbf{1st Moment ($\mathbb{E}[X]$) RMSE}} & \multicolumn{2}{c}{\textbf{2nd Moment ($\mathbb{E}[X^2]$) RMSE}} \\
\cmidrule{2-5}
 & \textbf{NUTS} & \ref{alg:DCS-HnR} & \textbf{NUTS} & \ref{alg:DCS-HnR}\\
\midrule
Centred Eight Schools & 0.117 & \textbf{0.038} & 1.814 & \textbf{0.354} \\
GARCH(1,1)             & 0.0428  & \textbf{0.0125} & 0.1069 & \textbf{0.0348} \\
\bottomrule
\end{tabular} 
\caption{Root Mean Square Error (RMSE) of the first and second moments for NUTS and~\ref{alg:DCS-HnR}. Samplers were allocated equal time budgets across 50 independent chains. Ground truth values are sourced from PosteriorDB~\cite{magnusson2024posteriordb}.}
\label{tab:posteriordb_rmse}
\end{table}

The summary metrics in Table~\ref{tab:posteriordb_rmse} demonstrate that even a simple, unoptimized implementation of ~\ref{alg:DCS-HnR} significantly outperforms NUTS on these pathological geometries. Because~\ref{alg:DCS-HnR} utilises a slice sampling mechanism, it  adapts to  the narrow funnel neck and the strict GARCH boundaries without requiring a globally tuned step size. Moreover, the contraction mapping in Algorithm~\ref{alg:DCS} mitigates the adverse effect of heavy-tails. Consequently, Algorithm~\ref{alg:DCS-HnR} achieves approximately a three-to-five-fold reduction in RMSE across both models, successfully recovering heavy-tailed variances without being trapped by hierarchical conditioning or dynamic constraints.

\subsection{Comparison with Stereographic samplers}
\label{subsec:comp_stereo}
Stereographic samplers~\cite{Yang24,grazzi2026sub}, designed for heavy-tailed target distributions, fall within the class of diffeomorphic contraction samplers defined above with the diffeomorphism $\CP_\beta:B(1)\stackrel{A}{\to}\mathbb{S}^{d}\setminus\{N\} \stackrel{B}{\to}\R^d$ factorising through the punctured sphere $\mathbb{S}^{d}\setminus\{N\}$ in $\R^{d+1}$ ($N$ denotes the North Pole in $\mathbb{S}^{d}$).
Here the function  $A$ maps radially the open disc $B(1)\subset \R^d$ onto a subset of  $\mathbb{S}^{d}$ in $\R^{d+1}$ (more specifically,  onto $\mathbb{S}^{d}\setminus\{N\}$ in~\cite{Yang24} and a spherical cap in  $\mathbb{S}^{d}$ in~\cite{grazzi2026sub}) and $B$ denotes the stereographic projection from $\mathbb{S}^{d}\setminus\{N\}$ onto $\R^d$.
Since both $A$ and $B$ only stretch the radial component of its input, it is evident that an appropriate scalar function $\phi_\beta:[0,1)\to[0,\infty)$ satisfying~\eqref{eq:beta_condition_for_phi} can be chosen so that $\CP_\beta$
equals $B\circ A$.
Once this identification has been made, the algorithms in~\cite{Yang24,grazzi2026sub} are given by Algorithm~\ref{alg:DCS} for a suitable choice of 
the automorphism $\mathcal{T}_\omega$ in~\eqref{eq:total_projection} and Markov kernel $Q$. Differently put, the numerical experiments in the present section can be cast as a comparison between two instances of Algorithm~\ref{alg:DCS}.

\subsubsection{Skewed student-t distributions}
\label{subsubsec:skewed_student_t}
In order to compare the performance of Algorithm~\ref{alg:DCS} with the recent Sub-Cauchy Sampler (SCS)~\cite{grazzi2026sub}, we replicate the experiment from paper~\cite[Sec~4.1]{grazzi2026sub}, where SCS showed a significant advantage over Gibbs sampling and HMC at targeting heavy-tailed distributions. The target chosen in\cite[Sec~4.1]{grazzi2026sub} is a skewed student-$t$ distribution~\eqref{eq:skewed_student_t}  in $d=100$ dimensions with $v=3$ degrees of freedom and  skewness parameter $\alpha\in\R^d$ with the first component equal to $100$, the second one to $-100$ and the rest equal to zero. To ensure a fair comparison, we employ an optimized sampling approach where the hyperparameters for both algorithms are calibrated prior to sampling via variational inference using the Adam optimizer~\cite{kingma2015adam}.

We fix the tail-decay parameter at $\beta=2.0$ for Algorithm~\ref{alg:DCS} and, following the SCS implementation in~\cite[Sec.~4.1]{grazzi2026sub}, set the latitude parameter to $l_0=1.10$. In each case, the parameter is chosen so that the corresponding algorithm is uniformly ergodic for the tail behaviour of the target distribution.

As described in Section~\ref{sec:optimization}, for Algorithm~\ref{alg:DCS} we minimize a Kullback--Leibler divergence to tune the location $\mu$, radius $R$, and the parameter $\delta$ of the M\"obius automorphism of the ball; see also~\eqref{eq:total_projection}. The resulting parameters $(\mu,R,\delta)$ are obtained using Algorithm~\ref{alg:VI}. The scale $R$, location $\mu$ and the longitude parameter $h_o$ in Algorithm~(SCS) are optimised using Adam following the procedure described in~\cite[Sec~4.1]{grazzi2026sub}.
To keep the comparison fair, in the class of Algorithms~\ref{alg:DCS} we sample using Algorithms~\ref{alg:DCS-RWM} and~\ref{alg:DCS-BW} as 
they are based on symmetric RWM kernels, a property shared by the SCS sampler in~\cite{grazzi2026sub}.

 Figure~\ref{fig:performance_mse} indicates that
Algorithms~\ref{alg:DCS-RWM} and~\ref{alg:DCS-BW} outperform SCS, achieving
lower mean squared error with respect to the ground truth for both the tail
probability and the expected norm. Algorithms~\ref{alg:DCS-RWM} and~\ref{alg:DCS-BW} are  more efficient than SCS on a
per-sample basis as each iteration is computationally less expensive because no
spherical moves are required. More precisely, by operating directly in the ball $B(R)$ in $\R^d$, as opposed to a subset of a sphere in $\R^{d+1}$, 
Algorithms~\ref{alg:DCS-RWM} and~\ref{alg:DCS-BW} avoid the trigonometric computations and boundary
stepping-out procedures on the sphere required by SCS in~\cite{grazzi2026sub}.

\begin{figure}[htbp]
    \centering
    \includegraphics[width=\textwidth]{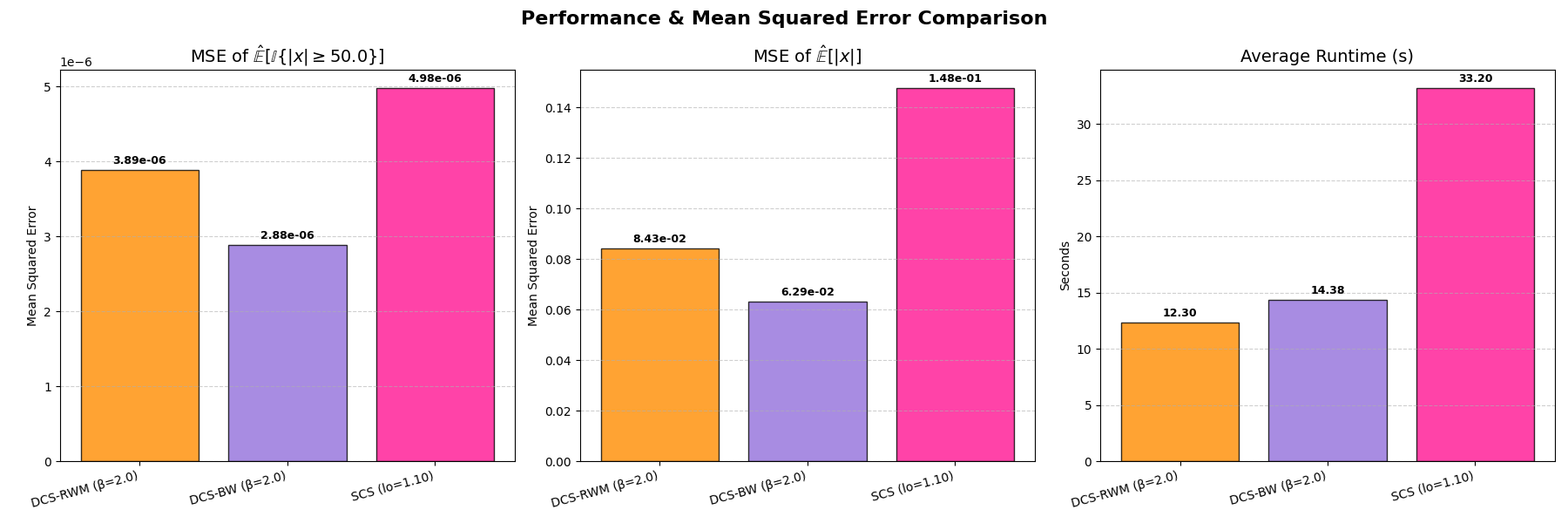}
    \caption{Comparison of average runtimes and mean squared errors across 20 independent chains (500,000 iterations each) for the sampling from the heavy-tailed skewed Student-$t$ target~\eqref{eq:skewed_student_t} (100-dimensional skewed Student-$t$ target with $v=3$ and skewness parameter $\alpha\in\R^d$ with the first component equal to $100$, the second one to $-100$ and the rest equal to zero). Samplers~\ref{alg:DCS-RWM} and~\ref{alg:DCS-BW}  achieve lower mean squared errors for  tail indicators and expected norms, while drastically reducing average runtime (in seconds) compared to the SCS algorithm~\cite{grazzi2026sub}.}
    \label{fig:performance_mse}
\end{figure}
\subsubsection{Sampling from targets with infinite first moment}
\label{subsubsec:infinite_first_moment}

Another feature of Algorithm~\ref{alg:DCS}  is its capacity to target distributions with arbitrarily heavy tails. To illustrate this, we design an experiment targeting a super-Cauchy multivariate  $t$-distribution with $v=0.5$ degrees of freedom (dimension $d=50$ and skewness parameter $\alpha\in\R^d$ is equal to zero). This yields a polynomial tail decay of approximately $\nu(|x| \ge r) \approx r^{-0.5}$. 

For this experiment, we run 50 independent chains for each algorithm, allocating a strict computational budget of 10.0 seconds per chain on the CPU. The~\ref{alg:DCS-BW} and~\ref{alg:DCS-RWM} use a tail parameter of $\beta=0.2$ to safely accommodate the extreme tails, while we deploy SCS sampler with latitude $l_0 = 1.1$ as a baseline, since this value has been observed in~\cite{grazzi2026sub} to yield best performance. Since the target $\pi$ is symmetric we choose  the location parameter to zero, fix longitude parameter to the origin in SCS~\cite{grazzi2026sub} and automorphism $\mathcal{T}_\omega:B(R)\to B(R)$ of the~\ref{alg:DCS} equal to the identity. 

The results, depicted in Figures~\ref{fig:superheavy_mse}, highlight a stark contrast in performance. Algorithms~\ref{alg:DCS-BW} and~\ref{alg:DCS-RWM} maintain uniform ergodicity for very heavy tailed distributions. Consequently, yielding mean squared errors an order of magnitude lower than SCS. On the other hand, as predicted by~\cite{grazzi2026sub}, uniform ergodicity breaks for SCS leading to poor performance for targets with infinite first moment. Finally, operating in the ball $B(R)$ (as opposed to a subset of a sphere in $\R^{d+1}$) allows~\ref{alg:DCS-BW} and~\ref{alg:DCS-RWM} to generate significantly more iterations within the fixed 10-second budget.

\begin{figure}[htbp]
    \centering
    \includegraphics[width=\textwidth]{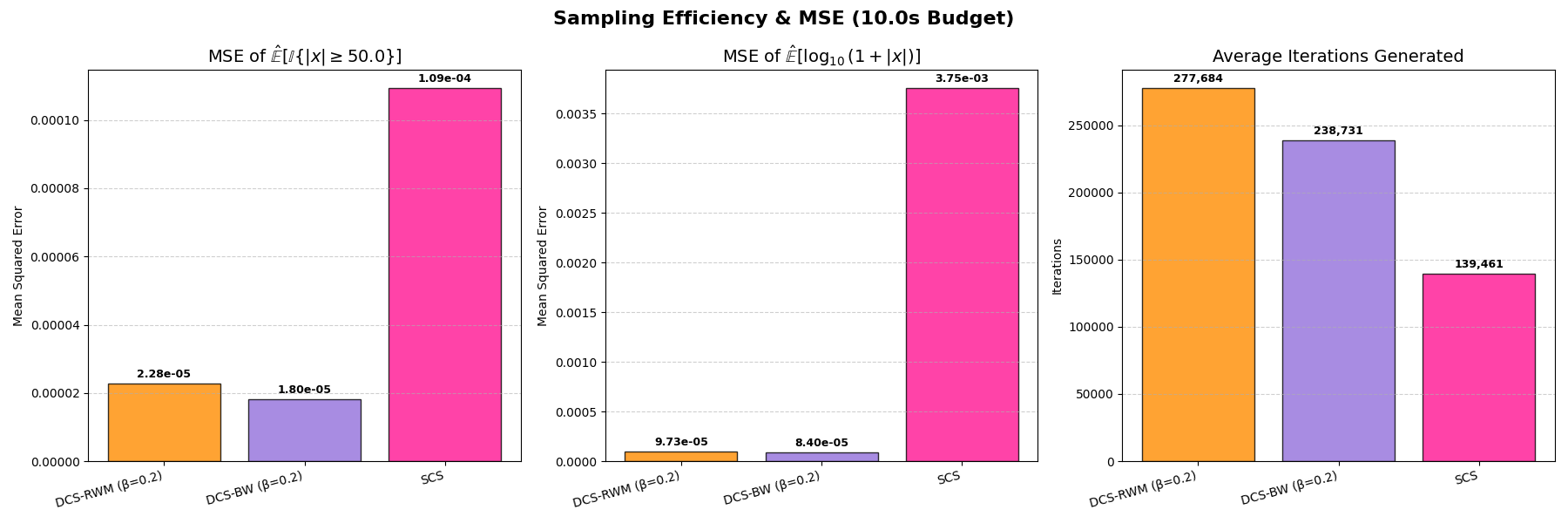}
    \caption{Average number of generated iterations and Mean Squared Error comparison across 50 independent chains (10.0 seconds time budget per chain). The properly tuned projected samplers ($\beta=0.2$) drastically outperform vanilla SCS in both estimation accuracy and computational throughput.}
    \label{fig:superheavy_mse}
\end{figure}

\subsubsection{Robust Bayesian binary regression: comparison of projected samplers}

In this subsection, we follow the experimental design of~\cite[Section~4.2]{grazzi2026sub} to analyse a Bayesian binary regression problem with perfectly separable data. When combined with weak prior information, separable data  induce heavy-tailed posteriors that typically cause standard Markov chain Monte Carlo algorithms to mix slowly or fail entirely. We compare the SCS and~\ref{alg:DCS-BW}. The tuning of the SCS parameters replicates exactly the variational methodology detailed in \cite[Section~4.2]{grazzi2026sub}. To ensure a fair computational comparison, we simulate 20 independent chains per sampler, allocating a fixed wall-clock time budget of 30 seconds per chain.

The Q-Q plots in Figure~\ref{fig:qq_plots} clearly illustrate the comparative performance of~\ref{alg:DCS-BW}. The CLT holds for the heavier prior in the case of~\ref{alg:DCS-BW}, while it struggles with SCS. On the other hand in the case of lighter prior, the confidence interval is thinner for the~\ref{alg:DCS-BW}.

\begin{figure}[htbp]
    \centering
    \includegraphics[width=13cm]{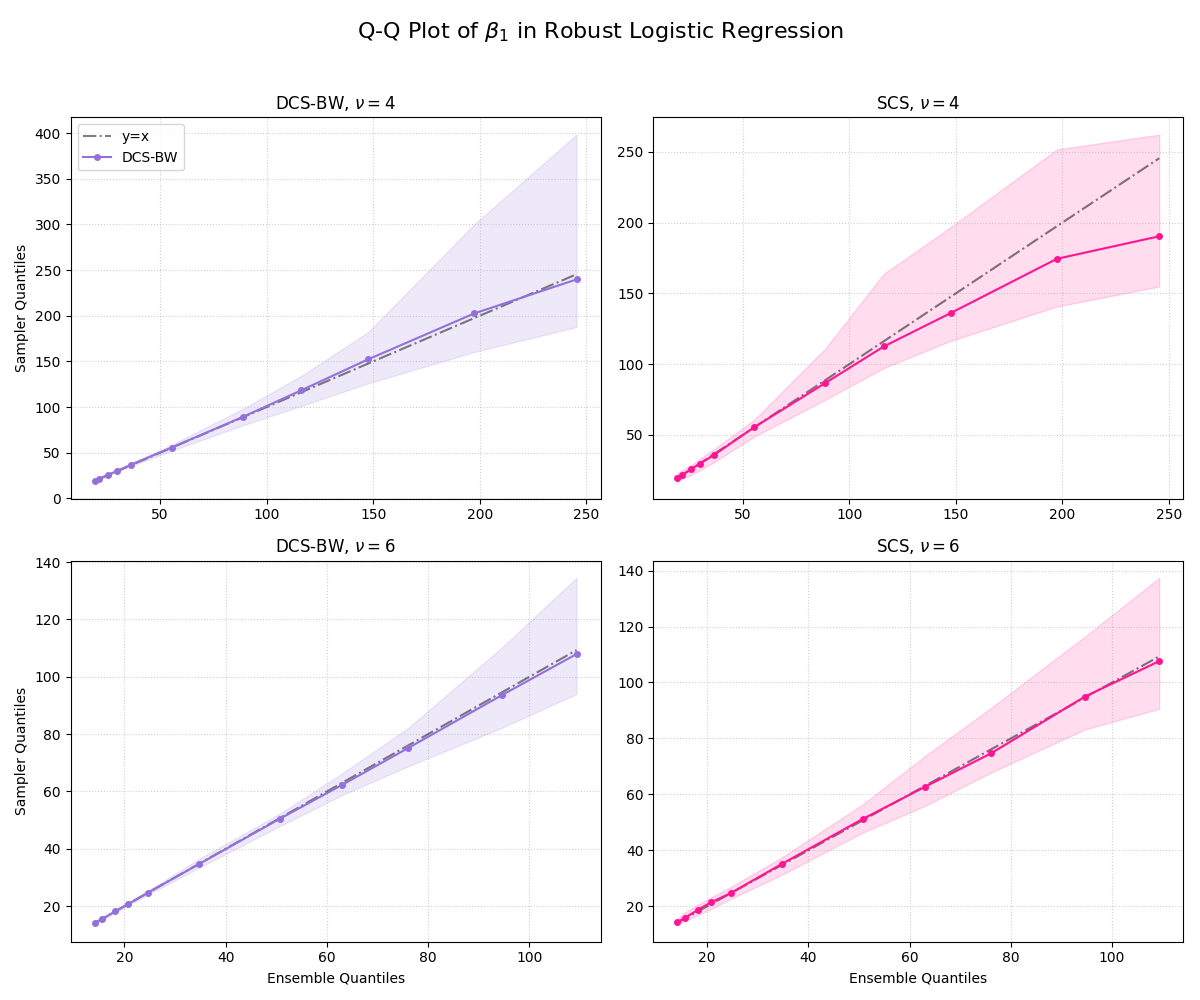}
    \caption{Q-Q plot of $\beta_1$ in robust logistic regression obtained from samples of algorithms~\ref{alg:DCS-BW} (left column) and SCS~\cite{grazzi2026sub} (right column). The prior is the Student-$t$ distribution with degrees of freedom equal to 4 (top), and 6 (bottom). Shaded regions show variation across 20 independent chains.}
    \label{fig:qq_plots}
\end{figure}

\section{Conclusions and open problems}
\label{sec:conclusions}

We introduced the Diffeomorphic Contraction Sampler~\ref{alg:DCS} for heavy-tailed target
distributions on $\R^d$. Algorithm~\ref{alg:DCS} pulls back the target density $\pi$ via the diffeomorphism~\eqref{eq:total_projection} to a density
$\pi_B$ on the Euclidean ball $B(R)$ of radius $R>0$.
The radial part of transformation~\eqref{eq:total_projection} is governed by the parameter $\beta\in(0,\infty)$.
We provide a simple sufficient tail-decay condition~\eqref{eq:unif_ergodicity_pi_condition} on $\pi$ in terms of $\beta$,
under which the resulting 
Algorithm~\ref{alg:DCS} 
is uniformly ergodic (see Theorem~\ref{thm:Uniform_ergodicity_unified} above) for standard
Markov kernels  on the ball $B(R)$ (see Appendix~\ref{app:Algs} below for  definitions of such kernels).

Under stronger geometric assumptions on the landscape induced by $\pi_B$, stated
in Section~\ref{subsec:non_asymptotic} above, the known mixing estimates for Markov chains on convex bodies yield  non-asymptotic bounds for Algorithm~\ref{alg:DCS} (see
Theorem~\ref{thm:dcs_bounds}).  The geometric assumptions of Section~\ref{subsec:non_asymptotic} can be satisfied for skewed and multi-modal initial targets $\pi$ on $\R^d$ by employing a suitable automorphism $\mathcal{T}_\omega:B(R)\to B(R)$ in~\eqref{eq:total_projection} (see Examples~\ref{ex:student_mixing} and~\ref{ex:multimodal_student_mixing} in Appendix~\ref{app:diffeo_convex} below), yielding first polynomial (in dimension) non-asymptotic mixing guarantees from cold start  (see
Theorem~\ref{thm:dcs_bounds}). 
Numerical results suggest that a rudimentary implementation of Algorithm~\ref{alg:DCS} outperforms state-of-the-art implementations of samplers such as NUTS~\cite{hoffman2014no, jax2018github}, for heavy tailed target distributions (see Section~\ref{sec:sim} above).

\smallskip

\noindent \underline{Open problems}. 
\textbf{(I) Gradient-based algorithms.} In the regime where Algorithm~\ref{alg:DCS} is uniformly ergodic, the transformed density $\pi_B$ may vanish at the boundary of the ball $B(R)$.  This may create difficulties for sampling algorithms on $B(R)$ such as ULA~\cite{MR3802303,brosse2017sampling}, HMC~\cite{livingstone2019geometric,MR4350970}, or the bouncy particle sampler~\cite{MR3911112} due to the unboundedness of the gradient. A natural open problem is to develop gradient-informed samplers on $B(R)$ that can handle this boundary singularity while retaining non-asymptotic convergence guarantees. Such methods could further improve the practical performance of Algorithm~\ref{alg:DCS} in regimes where random-walk based proposals are too diffusive.

\noindent \textbf{(II) Which automorphism $\mathcal T_\omega$ of $B(1)$ should one choose, and how?}
Examples~\ref{ex:student_mixing} and~\ref{ex:multimodal_student_mixing} show that a suitable automorphism $\mathcal T_\omega:B(1)\to B(1)$ can make the transformed target $\pi_B$ log-concave, yielding fast mixing of Algorithm~\ref{alg:DCS}. This leads to three tasks.

\noindent \textbf{(IIa)} Quantify theoretically the performance gain of Algorithm~\ref{alg:DCS} under a near-optimal automorphism $\mathcal T_\omega^{(o)}$. Natural criteria, motivated by Section~\ref{subsec:non_asymptotic}, include making $\pi_B$ close to a constant density on $B(R)$ in a suitable metric.

\noindent \textbf{(IIb)} Identify a tractable class of automorphisms of $B(1)$, possibly parametrised by deep-learning architectures, that approximates $\mathcal T_\omega^{(o)}$. The M\"obius transformation in Section~\ref{subsubsec:skewed_student_t} shows that this choice can be decisive. A richer class of computationally tractable automorphisms is needed for off-the-shelf use of Algorithm~\ref{alg:DCS} (note that only the existence of the automorphism used in Example~\ref{ex:multimodal_student_mixing} below is guaranteed by~\cite{bogachev2005triangular}).

\noindent \textbf{(IIc)} Prove non-asymptotic bounds for tuning parameters of the automorphism $\mathcal T_\omega$ within this class identified in (IIb). As Section~\ref{sec:optimization} explains, this is a natural transport-map problem for VI methods such as Adam~\cite{kingma2015adam}. Theoretical guarantees for this pre-processing step are therefore strongly motivated and highly desirable for demonstrating the robustness of Algorithm~\ref{alg:DCS}.

\noindent\textbf{(III) Ball-specific non-asymptotic bounds.}
The non-asymptotic bounds used in Section~\ref{subsec:non_asymptotic} are inherited from results for log-concave, or otherwise regular, densities on general convex bodies~\cite{lovasz2004hit,MR2309621,kook2024gaussian,jiang2024regularized}. In the present setting, however, the state space after contraction is the Euclidean ball $B(R)$ in $\R^d$, while the transformed density $\pi_B$ need not be close to uniform and may fail the regularity assumptions used in literature. It is thus natural to develop  non-asymptotic mixing bounds for sampling algorithms in $B(R)$. 
For instance, it is well known that Hit-and-Run possesses non-asymptotic mixing bounds from a cold start for any convex body, while Ball Walk does not~\cite{lovasz2004hit}. However, since $B(R)$ has no corners, fast mixing for a Ball-Walk algorithm from a cold start is known to hold for the uniform density on $B(R)$ by the conductance bound in~\cite{MR1608200}.

\section*{Funding}
This work was started while MB was a postdoc at Warwick Statistics, funded by the 
EPSRC grant EP/V009478/1 and is currently supported by CUHK-SZ start-up UDF01004230. AM was also supported by EP/V009478/1 and is supported by the EPSRC grant EP/W006227/1.

\bibliographystyle{amsplain}
\bibliography{lower.bib}

\appendix

\section{Proofs: convergence theory for Algorithm~\ref{alg:DCS}}
\label{app:proofs_non_asymptotic}
Recall the definition of the total variation distance of two probability measures $\nu_1,\nu_2$ on Borel sets $\mathcal{B}(U)$ for any metric set $U$:
\begin{equation}
\label{eq:TV_def}
    \| \nu_1 - \nu_2 \|_{\mathrm{TV}} \coloneqq \sup_{A \in \mathcal{B}(U)} |\nu_1(A) - \nu_2(A)|.
\end{equation}

We start with the following simple results, crucial for understanding of the performance of Algorithm~\ref{alg:DCS}.
\begin{prop}
\label{prop:diffeo_invariance}
Let $U, V \subseteq \mathbb{R}^d$ be open sets and let $\Phi: U \to V$ be a bi-measurable bijection. Let $\nu_1$ and $\nu_2$ be probability measures on $(U, \mathcal{B}(U))$. Then, the total variation distance is invariant under the push-forward by $\Phi$:
\begin{equation}
    \| \nu_1 - \nu_2 \|_{\mathrm{TV}} = \| \Phi_*\nu_1 - \Phi_*\nu_2 \|_{\mathrm{TV}},
\end{equation}
where the push-forward measure is defined  by $(\Phi_*\nu)(A) = \nu(\Phi^{-1}(A))$ for all Borel sets $A\in \mathcal{B}(V)$.
\end{prop}

\begin{proof}
The TV distance between the push-forward measures on $V$ is defined as:
\begin{equation*}
    \| \Phi_*\nu_1 - \Phi_*\nu_2 \|_{\mathrm{TV}} = \sup_{A \in \mathcal{B}(V)} |(\Phi_*\nu_1)(A) - (\Phi_*\nu_2)(A)|.
\end{equation*}
By the definition of the push-forward measure, $(\Phi_*\nu_i)(A) = \nu_i(\Phi^{-1}(A))$. Substituting this into the equation above yields:
\begin{equation*}
    \| \Phi_*\nu_1 - \Phi_*\nu_2 \|_{\mathrm{TV}} = \sup_{A \in \mathcal{B}(V)} |\nu_1(\Phi^{-1}(A)) - \nu_2(\Phi^{-1}(A))|.
\end{equation*}
Since $\Phi$ is a bijection, the mapping $A \mapsto \Phi^{-1}(A)$ is a bijection between the Borel $\sigma$-algebra $\mathcal{B}(V)$ and $\mathcal{B}(U)$. Therefore, taking the supremum over all sets $A \in \mathcal{B}(V)$ is equivalent to taking the supremum over all sets $E = \Phi^{-1}(A) \in \mathcal{B}(U)$.
\begin{align*}
    \sup_{A \in \mathcal{B}(V)} |\nu_1(\Phi^{-1}(A)) - \nu_2(\Phi^{-1}(A))| 
    &= \sup_{E \in \mathcal{B}(U)} |\nu_1(E) - \nu_2(E)| \\
    &= \| \nu_1 - \nu_2 \|_{\mathrm{TV}}.\qquad\qedhere
\end{align*}
\end{proof}

\begin{rem}
\label{rem_transform_application}
The result in Proposition~\ref{prop:diffeo_invariance} can be applied in our case since the composite map $\CP_{\text{total}}:B(R)\to \mathbb{R}^d$ defined in \eqref{eq:total_projection} is a diffeomorphism, and consequently a bijection, thereby preserving the TV distance. 
It is worth noting that while the Jacobian $J_{\CP_{\text{total}}}$ is required when transforming probability \textit{densities} (to ensure the integral sums to 1), it does not appear in the TV distance formula. 
\end{rem}

Proposition~\ref{prop:diffeo_invariance} implies Corollary~\ref{cor:convergence_equivalence} (with $\Phi=\RCtot$ in~\eqref{eq:total_projection}),  which in turn allows us to deduce the convergence bounds for Algorithms~\ref{alg:DCS} from the convergence of the chain on the ball $B(R)\subset \R^d$.

\begin{cor}
\label{cor:convergence_equivalence}
Let $(X_n)_{n \in \mathbb{N}}$ be a Markov chain on $\mathbb{R}^d$ with invariant measure $\nu$ with density proportional to $\pi:\R^d\to(0,\infty)$. Let $\RCtot: B(R) \to \R^d$ be the diffeomorphism defined in~\eqref{eq:total_projection} and let $(Z_n)_{n \in \mathbb{N}}$ be the transformed chain defined by $Z_n \coloneqq \RCtot^{-1}(X_n)$. Then  $(Z_n)_{n \in \mathbb{N}}$ is a Markov chain on $B(R)$ with invariant measure $\nu_B$ and density proportional to $\pi_{B}:B(R)\to(0,\infty)$ in~\eqref{eq:total_target}.
Moreover, for any $n \in\N$, we have
    \begin{equation}
    \label{eq:TV_equivalence_tranf}
        \| \P_x(X_n \in \cdot) - \nu \|_{\mathrm{TV}} = \| \P_{\RCtot^{-1}(x)}(Z_n \in \cdot) - \nu_{B} \|_{\mathrm{TV}}.
    \end{equation}
    Hence $(X_n)_{n \in \mathbb{N}}$  converges to $\nu$ in the total variation (TV) distance (see~\eqref{eq:TV_def} below for definition) if and only if the transformed process $(Z_n)_{n \in \mathbb{N}}$ converges to $\nu_{B}$.
\end{cor}


Corollary~\ref{cor:convergence_equivalence} offers a simple yet crucial insight into non-asymptotic theoretical analysis, as it provides a direct bridge between the convergence analysis of the process on the ball and that of the projected sampler. By mapping $\mathbb{R}^d$ to the ball~$B(R)$, we can leverage existing ergodicity results developed for samplers on convex sets (see e.g.,~\cite{lovasz2006fast,MR2309621, jiang2024regularized}).

\subsection{Proofs of non-asymptotic bounds} \label{app:mixing_proofs}In this subsection we prove our main results on non-asymptotic bounds (Theorem~\ref{thm:dcs_bounds}) by combining 
Corollary~\ref{cor:convergence_equivalence} with the non-asymptotic theory for sampling from convex bodies.

\begin{proof}[Proof of Theorem~\ref{thm:dcs_bounds}] By Corollary~\ref{cor:convergence_equivalence}, 
 the total-variation mixing time $\tau_X(\epsilon,\varsigma)$ in~\eqref{eq:mixing_def} of the  Markov chain  $X=(X_n)_{n \in \mathbb{N}}$, generated by Algorithm~\ref{alg:DCS} and started at $\varsigma$, is the same as the corresponding mixing time of the corresponding chain $Z = \RCtot^{-1}(X)$ on the ball $B(R)$. We now apply the non-asymptotic bounds  in~\cite{MR2309621,lovasz2004hit,MR3802303,kook2024gaussian} to the chain $Z$ on the convex set
 $B(R)$.
 
\noindent\underline{(a) Algorithm~\ref{alg:DCS-BW}} Assumption~\ref{assump:convex} implies that the density proportional to
$\pi_B$, extended by zero outside $B(R)$, is log-concave. Set $m_B\coloneqq \int_{B(R)}z\nu_B(z)dz$. By
Assumption~\ref{assump:isotropic}, we get
\begin{equation}
\label{eq:Var_bound_BW}
        \int_{B(R)}
        \left|z-m_B\right|^2\nu_B(z)dz\le \int_{B(R)}
        \left|z\right|^2\nu_B(z)dz 
        \leq Cd.
\end{equation}
Thus~\cite[Lem.~5.13]{MR2309621} implies that the density $\nu_B$ is
$a$-rounded with $a=(e\sqrt C)^{-1}$. The condition 
$\gamma\leq \frac{\eps^2}{2^{10}eH_0^2\sqrt{Cd}}$ on the step size $\gamma$ in Algorithm~\ref{alg:DCS-BW}
implies the step-size assumption in~\cite[Thm~2.2]{MR2309621}. Thus, by Assumption~\ref{assump:warm}, the conclusion of~\cite[Thm~2.2]{MR2309621} holds,  yielding (together with~\eqref{eq:Var_bound_BW}) the following inequalities:  
\begin{align*}       
\tau_X(\eps,\varsigma)&=
        \tau_Z(\eps,P_0)\leq
        10^{10}d
        \left(
        \int_{B(R)}
        |z-m_B|^2\nu_B(dz)
        \right)
        \gamma^{-2}
        \log\left(\frac{2H_0}{e\sqrt{C}\eps}\right)\\
        &\leq
        10^{10}Cd^2\gamma^{-2}
        \log\left(\frac{2H_0}{e\sqrt{C}\eps}\right).
\end{align*}

\noindent\underline{(b) Algorithm~\ref{alg:DCS-HnR}} (Warm start)
Assumptions~\ref{assump:convex} and~\ref{assump:isotropic} imply that  $\nu_B$ is a $C$-isotropic log-concave
density. Assumption~\ref{assump:warm} ensures
$P_0(A)\leq H_0\nu_B(A)$ for any Borel set $A\subset B(R)$. We can thus directly apply~\cite[Thm~2.3]{MR2309621} to obtain
\[
        \tau_X(\eps,\varsigma)
        \leq
        10^{30}C^4H_0^4d^3\eps^{-4}
        \log^3\left(\frac{2H_0}{\eps}\right).
\]
\noindent\underline{(c) Algorithm~\ref{alg:DCS-HnR}} (Cold start) 
Assumption~\ref{assump:convex} implies that $\pi_B$ is log-concave.
Assumption~\ref{assump:cold} verifies condition~\cite[(LS1)]{lovasz2006fast} with a second-moment bound $\overline r^2$ and
level-set radius $\underline r$. Moreover,~\ref{assump:cold} also  verifies condition~\cite[(LS2)]{lovasz2006fast}
with starting point $z$, distance-to-the-boundary constant $m$, and
approximate-mode constant $\varrho$. Consequently,
\cite[Cor.~1.2]{lovasz2006fast} gives
\[
        \tau_X(\eps,\varsigma)
        \leq
        10^{31}d^3
        \frac{\overline r^2}{\underline r^2}
        \left(\log\left(
        \frac{d\overline r^2}
        {\eps\underline r m\varrho}
        \right)\right)^5.
\]
\noindent\underline{(d) Algorithm~\ref{alg:DCS-LMC}} Assumptions~\ref{assump:convex}, \ref{assump:smooth}, and
\ref{assump:lipschitz} imply that the potential $U_B$ is convex, has
$M$-Lipschitz gradient, and satisfies $|\nabla U_B|\leq L$ on $B(R)$.
Thus the assumptions of~\cite[Thm~1]{MR3802303} hold for the projected
Langevin chain on $B(R)$, initialized at $Z_0=0$. By~\cite[Thm~1]{MR3802303} there exist universal constants
$C_{\textrm{LMC}},c_{\textrm{LMC}},C_{\textrm{step}},c_{\textrm{step}}>0$, such that Algorithm~\ref{alg:DCS-LMC} with step size $\eta = C_{\mathrm{step}}(\log f)^{c_{\mathrm{step}}} R^2/f^{12}$ generates a chain $X$ started at $0$ with mixing time bounded above by
\[\tau_X(\eps,\varsigma)
        \leq
        C_{\textrm{LMC}}R^6f^{12}
        (\log
        \left(R^6f^{12}\right))^{c_{\textrm{LMC}}},
        \qquad
        \text{where $f=\frac{\max\{d,RM,RL\}}{\eps}.$}
\]
\noindent\underline{(e) Algorithm~\ref{alg:DCS-DW}} 
Recall that the proposal covariance in
Algorithm~\ref{alg:DCS-DW} for $\lambda=0$ takes the form 
$\Sigma_0(z)
        =
        \frac{\gamma^2}{d}H(z)^{-1}$, where $H=d\nabla^2\varphi$ and $\varphi$ is the distance to the boundary of $B(R)$ given in~\eqref{eq:local_metric}.
By~\cite[Thm~3.5]{kook2024gaussian}, the metric $H=d\nabla^2\varphi$ is $\overline\nu$-Dikin-amenable for some $\overline\nu>0$ satisfying 
$\overline\nu\le c_0d$ for a constant $c_0>0$ independent of dimension (see~\cite[Defs~2.1 and 2.2]{kook2024gaussian} for the definition of $\overline\nu$-Dikin-amenable metric). Moreover,
\[
        \nabla^2\varphi(z)
        =
        \frac{2}{R^2-|z|^2}I_d
        +
        \frac{4zz^\top}{(R^2-|z|^2)^2}
        \succeq
        \frac{2}{R^2}I_d,\quad\text{implying $H(z)
        \succeq
        \frac{2d}{R^2}I_d$.}
\]
Assumptions~\ref{assump:convex} and~\ref{assump:smooth} imply
     $   0
        \preceq
        \nabla^2U_B(z)
        \preceq
        MI_d
        \preceq
        \frac{MR^2}{2d}H(z)$.
Hence, the potential $U_B$ is $0$-relatively strongly convex and
$\beta_{\textrm{DW}}$-relatively smooth in the metric $H$ (see~\cite[Sec.~C.1]{kook2024gaussian} for definitions of strong convexity and relative smoothness), where
$\beta_{\textrm{DW}} =MR^2/(2d)$.
Setting the step size in Algorithm~\ref{alg:DCS-DW}  at
\[
        \gamma
        =
        c_{\textrm{step}}
        \min\left\{
        1,
        \sqrt{\frac{2d}{MR^2}}
        \right\}
\]
for a sufficiently small constant
$c_{\textrm{step}}>0$ satisfies the  assumption in~\cite[Thm~3.1]{kook2024gaussian} requiring the step size to be proportional to
$1\wedge\beta_{\textrm{DW}}^{-1/2}$.

Finally, by Assumption~\ref{assump:warm} we have 
$dP_0/d\nu_B\le H_0$, implying that the bound on the starting distribution in~\cite[Thm~3.1]{kook2024gaussian} is at most $H_0$. Applying~\cite[Thm~3.1]{kook2024gaussian} 
with $\alpha=0$, $\beta=\beta_{\textrm{DW}}$ and
$\overline\nu\le c_0d$ gives
\[        \tau_X(\epsilon,\varsigma)\le
        C_{\textrm{DW}}'\,
        d\max\{1,\beta_{\textrm{DW}}\}
        \overline\nu
        \log\left(\frac{H_0}{\epsilon}\right)\le
        C_{\textrm{DW}}\,
        d^2
        \max\left\{
        1,\frac{MR^2}{2d}
        \right\}
        \log\left(\frac{H_0}{\epsilon}\right),
\]
for some constants 
$C_{\textrm{DW}}, C_{\textrm{DW}}'>0$ independent of dimension $d$ and  parameters $M,H_0,\epsilon,R$.
\end{proof}

\subsection{Proof of uniform ergodicity}
\label{app:uniform_ergodicity_proofs}
The proof of Theorem~\ref{thm:Uniform_ergodicity_unified} requires two-sided bounds on $\pi_B$, ensuring that the density is (a) bounded on $B(R)$ and (b) that it does not vanish away from the boundary of $B(R)$.  These properties of $\pi_B$ are used to control the acceptance probabilities in Algorithms~\ref{alg:DCS-HnR}, \ref{alg:DCS-BW}, \ref{alg:DCS-RWM}  and \ref{alg:DCS-DW}, 
implying that $\R^d$ is a small set for Algorithm~\ref{alg:DCS}, making it
 uniformly ergodic. 

By Lemma~\ref{lem:density_bounds}, properties (a) and (b) hold for any $\pi$ and $\beta>0$, such that  $x\mapsto \pi(x)|x|^{d+\beta}$ is a bounded function on $\R^d$. In the case of SCS algorithm~\cite{grazzi2026sub} (resp.~SPS sampler~\cite{Yang24}), the analogous assumption requires the function  $x\mapsto \pi(x)|x|^{d+1}$ (resp. $x\mapsto \pi(x)|x|^{2d}$) to be bounded.

\begin{lem}[Boundedness of the target density on $B(R)$]
\label{lem:density_bounds}
Let $\pi$ satisfy the condition in~\eqref{eq:unif_ergodicity_pi_condition} with $\beta>0$ and let $\pi_{B}$ in~\eqref{eq:total_target} be proportional to the density on the ball $B(R)$ induced by the pull-back via the $C^1$-diffeomorphism  $\RCtot(z) = \mu + R \cdot \CP_\beta(\mathcal{T}_\omega(z/R))$ given in~\eqref{eq:total_projection} for some parameters $\mu,R,\omega$. Then the following holds:\\
$\mathrm{(a)}$ There exists $M_1 < \infty$ such that $\pi_{B}(z) \le M_1$ for all $z \in B(R)$.\\
$\mathrm{(b)}$ For any $0<\ell < R$, there exists $m_\ell > 0$ such that $\pi_{B}(z) \ge m_\ell$ for $z \in B(\ell)=\{z'\in\R^d:|z'|<\ell\}$.
\end{lem}

\begin{proof}
We first prove $\mathrm{(a)}$. The density $\pi_B$ on the ball  $B(R)$, given in~\eqref{eq:total_target} and~\eqref{eq:total_jacobian}, satisfies
\begin{equation}
 \label{eq:pi_bounds_lemma}
    \pi_{B}(z)= A(z)\cdot J_{\mathcal{T}_\omega}(z/R), \quad\text{where $A(z)\coloneqq \pi(\CP_{\text{total}}(z)) J_{\CP_\beta}(\mathcal{T}_\omega(z/R))$, $z\in B(R)$.}
\end{equation}
Since $\mathcal{T}_\omega$ is a $C^1$-diffeomorphism on the closure $\overline{B}(1)=\{y\in\R^d:|y|\leq1\}$, it maps the boundary of $\overline{B}(1)$ bijectively and smoothly  onto itself. Consequently, 
we have
\begin{equation}
\label{eq:tau_to_one}
    |\mathcal{T}_\omega(z/R)| \to 1\qquad\text{as $|z|\uparrow R$.}
\end{equation}
Furthermore, the Jacobian $J_{\mathcal{T}_\omega}(s)$ is continuous on the compact set $\overline{B}(1)$, making it bounded. It is by~\eqref{eq:pi_bounds_lemma} thus sufficient to prove that $A:B(R)\to(0,\infty)$ is bounded. By continuity of $A$, it is in fact sufficient to prove the boundedness of $A$ on the annulus $B(R)\setminus B(R')$ for some $R'\in(0,R)$.

Pick $R'\in(0,R)$ sufficiently close to $R$, so that 
$\left|\CP_{\mathrm{total}}(z)\right|\geq 1$ and
$\left|\CP_\beta\!\left(\mathcal T_\omega\!\left(\frac{z}{R}\right)\right)\right|\geq 1$ for $z\in B(R)\setminus B(R')$.
Express $A(z)=A_1(z)\cdot A_2(z)\cdot A_3(z)$, where $A_1(z)\coloneqq \pi(\CP_{\text{total}}(z)) \left|\CP_{\text{total}}(z)\right|^{d+\beta}$,  
$$
A_2(z)\coloneqq \frac{|\CP_\beta(\mathcal{T}_\omega(z/R))|^{d+\beta}}{\left|\CP_{\text{total}}(z)\right|^{d+\beta}}, \quad A_3(z)\coloneqq
 \frac{J_{\CP_\beta}(\mathcal{T}_\omega(z/R))}{|\CP_\beta(\mathcal{T}_\omega(z/R))|^{d+\beta}},\quad \text{for $z\in B(R)\setminus B(R')$.}$$
Since $|\CP_{\text{total}}(z)|\to\infty$ as $|z|\uparrow R$ (by definition $\CP_{\text{total}}(z)= \mu + R \cdot \CP_\beta(\mathcal{T}_\omega(z/R))$ in~\eqref{eq:total_projection}, assumption~\eqref{eq:beta_condition_for_phi} and property~\eqref{eq:tau_to_one}), the tail condition in~\eqref{eq:unif_ergodicity_pi_condition} and continuity of $\pi$ imply that $A_1$ is bounded on $B(R)$.
By~\eqref{eq:total_projection} and~\eqref{eq:tau_to_one}, the function $A_2$ is bounded on $B(R)\setminus B(R')$. By the limit in~\eqref{eq:Jacobian_asymptotic_beta} and the property in~\eqref{eq:tau_to_one}, 
$A_3$ is also bounded on $B(R)\setminus B(R')$. This completes the proof of part~$\mathrm{(a)}$.

We now prove part~$\mathrm{(b)}$.
Pick $\ell < R$. The closed ball $\overline{B}(\ell)$ is a compact subset of the open ball $B(R)$. By assumption, the original target density $\pi$ is strictly positive and continuous on the compact set $\RCtot(\overline{B}(\ell))$ in $\R^d$. Furthermore, the Jacobians $w \mapsto |J_{\CP_\beta}(w)|$ and $w \mapsto |J_{\mathcal{T}_\omega}(w)|$ are strictly positive on the open domain $w\in B(1)$. Since $\pi_{B}$ in~\eqref{eq:pi_bounds_lemma} is the product of strictly positive, continuous functions, it is itself strictly positive and continuous on $B(R)$, implying part~(b). 
\end{proof}

\begin{lem}[$m$-step minorization on $B(R)$]
\label{lem:geometric_overlap}
Let assumptions of Theorem~\ref{thm:Uniform_ergodicity_unified} hold. For the Markov kernel $Q$ given by any of the 
Algorithms~\ref{alg:DCS-HnR}, \ref{alg:DCS-BW}, \ref{alg:DCS-RWM}  and \ref{alg:DCS-DW}  with covariance-floored proposal
$\Sigma_\lambda$ in~\eqref{eq:regularised:Dikin_H} for some $\lambda>0$,
there exist an integer $m \ge 1$ and a constant $\varepsilon > 0$, such that for all $z \in B(R)$ and any measurable $A \subseteq B(R)$ we have $Q^m(z, A) \ge \varepsilon \overline\nu(A)$, where $\overline\nu$ is the law of a uniform random vector taking values in  a ball centred at the origin. 
\end{lem}

\begin{proof}
By Lemma~\ref{lem:density_bounds}, for any $u\in(0,1)$, there exist constants $M_1 < \infty$ and  $m_{u} > 0$ such that 
\begin{equation}
\label{eq:z_1&z_2}
\text{$\pi_{B}(z_1) \le M_1$ and $\pi_{B}(z_2) \ge m_{u}$ for all $z_1 \in B(R)$ and $z_2 \in B(uR)$.}
\end{equation}

\noindent \underline{Proof for Algorithm~\ref{alg:DCS-HnR}.}
Let $\sigma$ denote the uniform probability measure on $\mathbb S^{d-1}$. For
$z\in B(R)$ and $\theta\in\mathbb S^{d-1}$, define $I(z,\theta):=\{t\in\mathbb R:\ z+t\theta\in B(R)\}$.
Then the transition kernel of the Hit-and-Run Markov chain equals 
\begin{equation}
\label{eq:HR-kernel-spherical}
Q(z_1,A)=\P_{z_1}(Z_1\in A)
=
\int_{\mathbb S^{d-1}}
\int_{I(z_1,\theta)}
\mathbf 1_A(z_1+t\theta)
\frac{\pi_B(z_1+t\theta)}
{\displaystyle \int_{I(z_1,\theta)} \pi_B(z_1+s\theta)\,\ud s}
\,\ud t\,\sigma(\ud\theta),
\end{equation}
for any $A\subseteq B(R)$ and $z_1\in B(R)$.
Now let $A\subseteq B(R/2)$. By~\eqref{eq:z_1&z_2} with $u=1/2$, we get
\[
\pi_B(z_1+t\theta)\ge m_{1/2}
\qquad\text{for all $t\in\R$ satisfying } z_1+t\theta\in A,
\]
while
\[
\int_{I(z_1,\theta)} \pi_B(z_1+s\theta)\,\ud s
\le 2R\,M_1,
\]
since the chord length is of length at most \(2R\) and \(\pi_B\le M_1\) on \(B(R)\). Hence
from~\eqref{eq:HR-kernel-spherical},
\[
\P_{z_1}(Z_1\in A)
\ge
\frac{m_{1/2}}{2RM_1}
\int_{\mathbb S^{d-1}}
\int_{I(z_1,\theta)}
\mathbf 1_A(z_1+t\theta)\,\ud t\,\sigma(\ud\theta).
\]

Since $B(R)=\{z_1+r\theta:\theta\in\mathbb{S}^{d-1}, r\in I(z_1,\theta)\cap\R_+\}$,
using polar coordinates $ (r,\theta)\mapsto z=z_1+r\theta$ centred at $z_1$ (with Jacobian $|z-z_1|^{d-1}$), for an appropriate constant $C_d>0$ we obtain
\[
\P_{z_1}(Z_1\in A)
\ge
C_d \frac{m_{1/2}}{2RM_1}
\int_A |z-z_1|^{1-d}\,\ud z.
\]
Since $z_1\in B(R)$ and $A\subseteq B(R/2)$, we have $|z-z_1|\le |z|+|z_1|\le 2R$ for all
\(z\in A\). Thus
\[
\P_{z_1}(Z_1\in A)
\ge
C_d (2R)^{-d}\frac{m_{1/2}}{M_1}\,\mathrm{Leb}(A)\quad\text{for all measurable $A\subset B(R/2)$,}
\]
implying the lemma (with $m=1$ and $\overline\nu$ equal to the uniform law on $B(R/2)$) for Algorithm~\ref{alg:DCS-HnR}.

\noindent \underline{Proof for Algorithm~\ref{alg:DCS-DW}.} In the case of Dikin Walk, the proposal density equals $q(z_1, z_2) = \mathcal{N}(z_2; z_1, \Sigma_\lambda(z_1))$, where the covariance matrix is defined as $\Sigma_\lambda(z) = \frac{\gamma^2}{d}(H(z)^{-1} + \lambda I_d)$. Since $H$ is a  positive-definite local metric on the ball $B(R)$ given by~\eqref{eq:local_metric}, it is bounded from below. Moreover since $\lambda > 0$, the spectrum of $\Sigma_\lambda(z)$ is contained in the interval $[c_{\min}, c_{\max}]$
for all $z\in B(R)$ and  $0 < \gamma^2\lambda/d=: c_{\min} < \gamma^2(\lambda + R^2/(2d))/d=:c_{\max} < \infty$.
Hence, $c_{\min} I_d \preceq \Sigma_\lambda(z) \preceq c_{\max} I_d$ for all $z \in B(R)$. Since $|z_1 - z_2| \le 2R$, the Gaussian proposal density is uniformly bounded from below by
\[
    q(z_1, z_2) \ge \frac{1}{(2\pi c_{\max})^{d/2}} \exp\left(-\frac{2R^2}{c_{\min}}\right) \eqqcolon \delta_q > 0.
\]
Accounting for the probability of $1/2$ associated with the lazification step in Algorithm~\ref{alg:DCS-DW}, 
for any measurable $A\subset B(R)$
the  one-step  transition probability $Q(z_1,\cdot)=\P_{z_1}(Z_1\in \cdot)$  satisfies
\[
  \P_{z_1}(Z_1\in A) \ge \frac{1}{2} \int_A \min\left(q(z_1, z_2), \frac{\pi_{B}(z_2)}{\pi_{B}(z_1)} q(z_2, z_1)\right)\ud z_2.
\]
Applying the uniform bounds $q(\cdot, \cdot) \ge \delta_q$ on $B(R)\times B(R)$ and  $\pi_B \le M_1<\infty$ on $B(R)$, $\pi_B \ge m_{1/2}>0$ on $B(R/2)$ from~\eqref{eq:z_1&z_2},  for all measurable $A\subset B(R/2)$ we obtain
\[
  \P_{z_1}(Z_1\in A) 
 \ge \frac{1}{2} \min\left(\delta_q, \frac{m_{1/2}}{M_1} \delta_q \right) \mathrm{Leb}(A),
\]
where  $\min(\delta_q, \frac{m_{1/2}}{M_1} \delta_q)>0$, concluding the proof for Algorithm~\ref{alg:DCS-DW} (with $m=1$ and $\overline\nu$ the uniform law on $B(R/2)$).

In the remainder of the proof,
denote  by
$\alpha(z_1, z_2)\coloneqq \min(1, \pi_{B}(z_2) / \pi_{B}(z_1))$  the acceptance probability   for any $z_1,z_2\in B(R)$. By~\eqref{eq:z_1&z_2}, for any $u\in(0,1)$, the acceptance probability is uniformly bounded below on the set $B(R)\times B(uR)$ by $\alpha\geq \alpha_{\min,u} \coloneqq m_{u} / M_1 > 0$.

\noindent \underline{Proof for Algorithm~\ref{alg:DCS-RWM}.} The proposal density $q(z_1, z_2)$ is Gaussian and hence   there exists $\delta_q > 0$ such that $q(z_1, z_2) \ge \delta_q$ for all $(z_1,z_2)\in B(R) \times B(R)$. It follows that
for any measurable $A\subset B(R/2)$, the one-step transition probability $Q(z_1,\cdot)=\P_{z_1}(Z_1\in \cdot)$ of the chain satisfies
\[
 \P_{z_1}(Z_1\in A) 
 \ge   \int_A q(z_1, z_2) \alpha(z_1, z_2) \ud z_2 \ge \delta_q \alpha_{\min,1/2} \mathrm{Leb}(A)
\quad\text{for any $z_1 \in B(R)$,}\]
implying the lemma (with $m=1$ and $\overline\nu$ equal to the uniform law on $B(R/2)$) for Algorithm~\ref{alg:DCS-RWM}.

\noindent \underline{Proof for Algorithm~\ref{alg:DCS-BW}.} Ball walk may need several steps to get from an arbitrary point $z\in B(R)$ to a neighbourhood of the origin. Recall that for a state $z \in B(R)$, the proposal distribution of the ball walk is uniform on the set $z+B(\gamma)$ for some fixed $\gamma > 0$. Since the maximum distance between any two points in $B(R)$ is $2R$, we consider paths of length $m \coloneqq \lceil 4R/\gamma \rceil$, where $\lceil x\rceil$ denotes the smallest integer greater or equal to $x\in\R$. 

Denote by $V_\gamma\coloneqq \textrm{vol}_d(B(\gamma))$ the volume of $B(\gamma)$. The one-step transition probability 
$Q(z,\cdot)=\P_z(Z_1\in \cdot)$
of the Ball-Walk chain satisfies
\begin{equation}
\label{eq:one_step_lower_bound_BW}
       \P_z(Z_1\in  A)\ge p_{\min,u} \mathrm{Leb}(A),\quad\text{where $p_{\min,u}\coloneqq \alpha_{\min,u}/V_\gamma$,}
\end{equation}
for any $z\in B(R)$, $u\in(0,1)$ and measurable $A\subset (z+B(\gamma))\cap B(uR)$ (recall $\alpha_{\min,u}=\frac{m_u}{M_1}$ in~\eqref{eq:z_1&z_2}).

 To establish the $m$-step minorization, consider 
 an arbitrary starting state $z\in B(R)$ and let $c_i\coloneqq  z(1-  \frac{i}{m})$ for $i = 1, \dots, m$, implying $|c_i|\leq R(1- 1/m)$ and that the distance between consecutive points is bounded above by $|c_i - c_{i-1}| = |z|/m \le\gamma/4$ (here $c_0\coloneqq z$). Around each centre $c_i$, for $i=1,\ldots,m$, we define an intermediate balls $B_i \coloneqq c_i+B(r)$ with radius $r \coloneqq\min(\gamma/4, R/(4m))$. For any state $v_i \in B_i$, we have  $|v_i| < |c_i| + r \leq R(1- 1/m)+R/(4m)  < R - R/(4m)$, implying $B_i\subset B(uR)$ for $u \coloneqq 1 - 1/(4m) \in (0, 1)$ 
and all $i=1,\ldots,m$.

The distance between any two points $v_i \in B_i$ and $v_{i+1} \in B_{i+1}$ (recall $v_0= z$ and $v_m \in B_M=B(r)$) satisfies
\[
    |v_i - v_{i+1}| \le |v_i - c_i| + |c_i - c_{i+1}| + |c_{i+1} - v_{i+1}| \le r + \frac{\gamma}{4} + r \le \frac{3\gamma}{4} < \gamma.
\]
Since  $B_{i+1}\subset v_i+B(\gamma)$, the transition from $v_i$ to  $B_{i+1}$ has positive probability. Since all $v_i$, for $i=1,\ldots,m$, are contained in $B(uR)$, by~\eqref{eq:one_step_lower_bound_BW} the one-step transition probability satisfies 
\begin{equation}
\label{eq:BW_transiton_one_step}
\P_{v_{i-1}}(Z_1\in B_i)\geq p_{\min,u}V_r\quad\text{for any $i=1,\ldots,m$, where $V_r\coloneqq \text{vol}_d(B(r)).$}
\end{equation}
The $m$-step transition kernel $Q^m(z,\cdot)=\P_z(Z_m\in \cdot)$ is by the Chapman-Kolmogorov equation,~\eqref{eq:one_step_lower_bound_BW} and~\eqref{eq:BW_transiton_one_step} bounded below by
\begin{align*}
\P_z(Z_m\in A)
    &\ge \int_{B_1\times\cdots \times B_{m-1}}  \P_{v_0}(Z_1\in \ud v_1)  \cdots \P_{v_{m-2}}(Z_1\in \ud v_{m-1} )
   \P_{v_{m-1}}(Z_1\in A\cap B(r))\\ 
   & \ge 
    p_{\min,u}^m  V_r^{m-1} \mathrm{Leb}(A\cap B(r)),
\end{align*}
for any measurable set $A\subset B(R)$, establishing the lemma for Algorithm~\ref{alg:DCS-BW} (with $\overline\nu$ the uniform law on $B(r)$ and $m= \lceil 4R/\gamma \rceil$ defined above). 
\end{proof}

\begin{proof}[Proof of Theorem~\ref{thm:Uniform_ergodicity_unified}]
Let $P$ denote the one-step transition kernel of the chain $X=(X_n)_{n\in\N}$. Put differently, we have
$P^n(x,\cdot)\coloneqq\P_x(X_n\in\cdot)$ for all $x\in\R^d$.
A set $C\subset \R^d$ is a \textit{small set} for a chain $X$ if there exist $m\in\N$, $\epsilon>0$ and a probability measure $\overline\mu$ on measurable sets in $\R^d$ such that $$P^m(x, \cdot)\geq \epsilon \overline\mu(\cdot)\quad\text{ for all $x\in C$.}$$
Since $X=\RCtot(Z)$, by Lemma~\ref{lem:geometric_overlap}, the push-forward onto $\R^d$ of the measure $\overline\nu$ on $B(R)$  (via $\RCtot:B(R)\to\R^d$) implies that 
under the assumptions of Theorem~\ref{thm:Uniform_ergodicity_unified}, the entire state space $C=\R^d$ is small for  the chain $X$ defined by
Algorithm~\ref{alg:DCS} for some $\epsilon>0$ and the push-forward probability measure $\overline\mu:=(\RCtot)_*\overline\nu$:
\begin{align*}
P^m(x,A)
&=
Q^m\left(
\RCtot^{-1}(x),
\RCtot^{-1}(A)
\right) 
\geq
\varepsilon\,
\overline\nu\left(\RCtot^{-1}(A)\right)
=
\varepsilon\,\overline\mu(A),\quad 
A\in\mathcal B(\R^d).
\end{align*}
After decreasing $\varepsilon$ if necessary, assume that
$\varepsilon\in(0,1)$ and define the kernel
$\widetilde P(x,A)
:=
\frac{P^m(x,A)-\varepsilon\overline\mu(A)}
     {1-\varepsilon}$.
The minorization implies that $\widetilde P$ is a Markov kernel and
$P^m(x,\cdot)
=
\varepsilon\overline\mu(\cdot)
+
(1-\varepsilon)\widetilde P(x,\cdot)$.
Consequently, for any probability measures $\eta_1$ and $\eta_2$,
\begin{align}
\label{eq:Doeblin_contraction}
\left\|\eta_1P^m-\eta_2P^m\right\|_{\mathrm{TV}}
=
(1-\varepsilon)
\|\eta_1\widetilde P-\eta_2\widetilde P\|_{\mathrm{TV}}
\leq
(1-\varepsilon)
\left\|\eta_1-\eta_2\right\|_{\mathrm{TV}},
\end{align}
since every Markov kernel  contracts total variation.

Since $\nu$ is invariant for $P$, iterating
\eqref{eq:Doeblin_contraction} yields
\[
\|P^{km}(x,\cdot)-\nu\|_{\mathrm{TV}}
\leq
(1-\varepsilon)^k
\left\|\delta_x-\nu\right\|_{\mathrm{TV}}
\leq
(1-\varepsilon)^k\quad\text{for all $k\in\N$.}
\]
For any $n\in\N$, write
$n=km+j$,
$k=\left\lfloor\frac{n}{m}\right\rfloor$
and
$0\leq j<m$.
Another application of total-variation contraction gives
\[
\left\|P^n(x,\cdot)-\nu\right\|_{\mathrm{TV}}
\leq
\|P^{km}(x,\cdot)-\nu\|_{\mathrm{TV}}
\leq
(1-\varepsilon)^{\lfloor n/m\rfloor}.
\]
Therefore, setting
$\rho:=(1-\varepsilon)^{1/m}\in(0,1)$ and
$C:=(1-\varepsilon)^{-1}$,
we obtain
\[
\left\|\P_x(X_n\in\cdot)-\nu\right\|_{\mathrm{TV}}
=
\left\|P^n(x,\cdot)-\nu\right\|_{\mathrm{TV}}
\leq
C\rho^n\quad\text{for every $x\in\R^d$ and $n\in\N$.}
\]
Hence $X$ is uniformly ergodic.
\end{proof}

\section{Skewed and multimodal Student-$t$ laws satisfy fast-mixing assumptions}
\label{app:diffeo_convex}

Theorem~\ref{thm:dcs_bounds} above gives non-asymptotic convergence bounds for
Algorithm~\ref{alg:DCS} under geometric assumptions on the transformed density
$\pi_B$ in~\eqref{eq:total_target}. The key assumption in the theorem is the
convexity of the transformed potential $U_B=-\log \pi_B$, namely
Assumption~\ref{assump:convex}. The purpose of this appendix is to provide examples of  targets $\pi$ on $\R^d$ that are transformed by~\eqref{eq:total_projection} into a
density $\pi_B$ with a convex potential. In this appendix we work with $\CP_\beta$ given in~\eqref{eq:compact_projection_first}. We start with Student-$t$ targets and then  extend the method to the  multimodal settings.

Recall that by definition of the Student-$t$ distribution has a density on $\R^d$ proportional to
\begin{equation}
\label{eq:student_def}
\pi(x)
=
\left(
1+\frac{1}{v}(x-\mu)^\top\Sigma^{-1}(x-\mu)
\right)^{-(v+d)/2},
\qquad x\in\R^d,
\end{equation}
where $v>0$ and $\Sigma\in \R^{d\times d}$ is a positive definite matrix.

\begin{example}
\label{ex:student_mixing}
For every $v\geq2$ and positive definite $\Sigma$, the Student-$t$ target in~\eqref{eq:student_def} admits a diffeomorphism
$\RCtot:B(R)\to \R^d$ with $R^2=d+v$ and $\beta = 2$ such that the transformed potential
$U_B=-\log\pi_B$ is convex on $B(R)$ and the transformed probability measure $\nu_B$ is $1$-isotropic (i.e.,~\ref{assump:isotropic} holds with $C=1$). Moreover, Assumption~\ref{assump:cold} holds for the initial point $z=0$
with
\begin{equation}
\label{eq:Cold_params}  
        \underline r=\frac{\sqrt d}{8},
        \qquad
        \overline r^2=d,
        \qquad
        m\geq1,
        \qquad
        \varrho=1.
\end{equation}
Consequently, by Theorem~\ref{thm:dcs_bounds} above, Algorithms~\ref{alg:DCS-BW} and~\ref{alg:DCS-HnR} mix in
$O(d^3)$ from a warm start, while
Algorithm~\ref{alg:DCS-HnR} mixes in $O(d^3)$ from the cold start.
\end{example}

\begin{rem}[Improved Ball-Walk bound under roundedness]
\label{rem:ball_walk_roundedness}
The Ball-Walk bound in Theorem~\ref{thm:dcs_bounds} is stated in terms of
$C$-isotropy in order to use a common and easily verifiable geometric
assumption throughout the paper. Alternatively, the Ball-Walk result may be
formulated using the notion of $a$-roundedness; see
\cite[p.~310]{lovasz2006fast}. In this formulation, for an isotropic target,
the mixing bound is of order
\[
        d^2\gamma^{-2}
        \log\left(\frac{H_0}{\epsilon}\right)
        \quad\text{for step size of order}\quad
        \gamma
        \asymp
        \frac{a\epsilon^2}{H_0^2\sqrt d}.
\]
For a general isotropic log-concave target, the roundedness parameter is only
guaranteed to be of constant order, which recovers the mixing time complexity of
$d^3$. In contrast, for the transformed Student-$t$ targets in
Example~\ref{ex:student_mixing}, the radial geometry of the level sets gives
\[
        a\ge\frac{\sqrt d}{8}.
\]
Consequently, $\gamma$ may be chosen independently of $d$, and the
warm-start Ball-Walk mixing bound improves to $d^2$.

We retain the $C$-isotropy formulation in
Theorem~\ref{thm:dcs_bounds} for ease of presentation and regard the
quadratic bound as a sharper, family-specific consequence of the stronger
roundedness property.
\end{rem}

\begin{proof}[Proof of Examples~\ref{ex:student_mixing}]
Recall the diffeomorphism $\CP_2: B(1)\to \R^d$ in~\eqref{eq:compact_projection_first} and set $R^2=d+v$.
Define the automorphism $\mathcal T_\omega: B(1)\to B(1)$ by
$$
\mathcal T_\omega(q)
=
\CP_2^{-1}\left(\frac{\sqrt{v}}{R}\Sigma^{1/2}\CP_2(q)\right),
\qquad q\in B(1).
$$
Equivalently,
$$
\mathcal T_\omega(q)
=
\frac{(\sqrt{v}/R)\Sigma^{1/2}q}
{\sqrt{1-|q|^2+\left|(\sqrt{v}/R)\Sigma^{1/2}q\right|^2}},
\qquad \text{with inverse} \qquad
\mathcal T_\omega^{-1}(y)
=
\frac{(R/\sqrt{v})\Sigma^{-1/2}y}
{\sqrt{1-|y|^2+\left|(R/\sqrt{v})\Sigma^{-1/2}y\right|^2}}.
$$
For $q=z/R$ (i.e., $z\in B(R)$), the total diffeomorphism in~\eqref{eq:total_projection} equals
$$
\RCtot(z)
=
\mu+\CP_2(\mathcal T_\omega(z/R))R
=
\mu+\sqrt{v}\,\Sigma^{1/2}\CP_2(z/R).
$$
Hence
$$
(\RCtot(z)-\mu)^\top\Sigma^{-1}(\RCtot(z)-\mu)
=
v|\CP_2(z/R)|^2.
$$
The definition of $\pi_B$ in~\eqref{eq:total_target} yields

$$
\pi_B(z)
=
\left(1+|\CP_2(z/R)|^2\right)^{-(v+d)/2}
\left| J_{\CP_2}(z/R)\right|.
$$
Using the identities
$$
1+|\CP_2(z/R)|^2
=
\frac{1}{1-|z|^2/R^2},
\qquad
\left| J_{\CP_2}(z/R)\right|
=
\left(1-\frac{|z|^2}{R^2}\right)^{-(1+d/2)},
$$
we obtain
\begin{equation}
\label{eq:tranformed_student_ex}
\pi_B(z)
=
\left(1-\frac{|z|^2}{R^2}\right)^{v/2-1}.
\end{equation}
Therefore
$$
U_B(z)
=
C_0-
\left(\frac{v}{2}-1\right)
\log\left(1-\frac{|z|^2}{R^2}\right)
 ~\&~
\nabla^2 U_B(z)
=
\frac{v-2}{R^2}
\left[
\frac{I}{1-|z|^2/R^2}
+
\frac{2zz^\top/R^2}{(1-|z|^2/R^2)^2}
\right]
\succeq0.
$$
Thus, Assumption~\ref{assump:convex} holds. Moreover, $\pi_B$ is exactly isotropic. Indeed, if $Z\sim \nu_B$, then
$\mathbb E |Z|^2=R^2\frac{d}{d+v}.$
Since $R^2=d+v$, we get $\mathbb E |Z|^2=d$. By radial symmetry, $\E[ZZ^\top] =  \E |Z|^2/dI_d  =  I_d.$
Thus, Assumption~\ref{assump:isotropic} holds with $C=1$.  

It remains to verify Assumption~\ref{assump:cold}. From~\eqref{eq:tranformed_student_ex}, we obtain
\[
        \nu_B(z)
        =
        \frac{\Gamma((d+v)/2)}
        {\pi^{d/2}R^d\Gamma(v/2)}
        \left(1-\frac{|z|^2}{R^2}\right)^{v/2-1},
        \qquad z\in B(R).
\]
Since $v\ge2$, the density $\nu_B$ is radially symmetric and non-increasing
in $|z|$. Hence every non-empty level set $\mathcal L(c)$ is a ball
centred at the origin. Since $\nu_B(z)\le \nu_B(0)$, we have
\[ \nu_B(B(\sqrt d/8)) \le
        \nu_B(0)\operatorname{vol}_d(B(\sqrt d/8))=
        \frac{\Gamma((d+v)/2)}
        {\Gamma(v/2)\Gamma(d/2+1)}
        \left(\frac{d}{64(d+v)}\right)^{d/2}\le
        \frac{(d/128)^{d/2}}{\Gamma(d/2+1)}
        <1/8,
\]
where the second inequality follows from Wendel's inequality $\Gamma(x+a)\leq \Gamma(x)(x+a)^a$ with $a= d/2$, $x=v/2$ (see e.g.~\cite[Eq.~2.2]{luo2012bounds}), and the 
the last inequality holds by  $\Gamma(d/2+1) \ge (d/(2e))^{d/2}$ for $d\ge2$ (see~\cite{batir2017bounds}).
Consequently, every level set
$\mathcal L(c)$ satisfying $\nu_B(\mathcal L(c))\ge 1/8$ contains $B(\sqrt d/8).$ Thus, the level-set condition in Assumption~\ref{assump:cold} holds with
$\underline r=\sqrt d/8$.
Moreover, by isotropy we already obtained 
$      \int_{B(R)}|z|^2\,\nu_B(z)dz=d$, implying $\overline r^2=d$.

The density $\nu_B$ attains its maximum at $0$. Hence
$\nu_B(0)=\sup_{z\in B(R)}\nu_B(z)$, implying $\varrho=1$.
For $v>2$, the level set of $\pi_B$ at the value $\pi_B(0)/2$ is by~\eqref{eq:tranformed_student_ex}  equal to
\[
       \mathcal L(\pi_B(0)/2)
        =
        B\left(
        R\sqrt{1-2^{-2/(v-2)}}
        \right).
\]
For $v=2$, the function $\pi_B$  is constant, implying
\[
        \mathcal L(\pi_B(0)/2)=B(R).
\]
Therefore the distance of the origin from the boundary of the level set $L(\nu_B(0)/2)$ equals
\[
       m= \operatorname{dist}
        \left(
        0,\partial\mathcal L(\nu_B(0)/2)
        \right)
        =
        \begin{cases}
        R, & v=2,\\[1mm]
        R\sqrt{1-2^{-2/(v-2)}}, & v>2.
        \end{cases}
\]
Since this distance is at least $1$ for every $d\ge1$ and $v\ge2$, we have
 $m\geq1$.
\end{proof}

\begin{example}[Multimodal Student-$t$ mixture]
\label{ex:multimodal_student_mixing}
Let $v\geq2$ and $a_1,,\ldots, a_m\in \R^d$. Consider the multimodal Student-$t$ mixture distribution
$$
\pi(x)
=
\sum_{j=1}^m
w_j c_{v,d}
\left(
1+\frac{|x-a_j|^2}{v}
\right)^{-(v+d)/2},
\qquad
w_j>0,
\qquad
\sum_{j=1}^m w_j=1.
$$
There exists a diffeomorphism
$\RCtot: B(R)\to\R^d$ such that the transformed potential
$U_B=-\log\pi_B$ is convex on $B(R)$. Consequently,
Algorithms~\ref{alg:DCS-HnR} and~\ref{alg:DCS-BW} mix in $O(d^3)$ from a warm
start, while Algorithm~\ref{alg:DCS-HnR} mixes in $O(d^3)$ from a cold start.
\end{example}

\begin{proof}[Proof of Example~\ref{ex:multimodal_student_mixing}]
Let $\pi_0$ be the centred symmetric Student-$t$ density with parameter
$v$ and $\Sigma = \Id$ in~\eqref{eq:student_def}. Since both $\pi$ and $\pi_0$ are smooth and strictly positive, the
Knothe--Rosenblatt transport theorem~\cite[Prop.~2.5]{bogachev2005triangular} yields a smooth triangular diffeomorphism
$\Psi:\R^d\to\R^d$ such that
$$
\pi(x)
=
\pi_0(\Psi(x))|J_\Psi(x)|.
$$
Take $R^2=d+v$, use $\CP_2$, and define
$$
\mathcal T_\omega(q)
=
\CP_2^{-1}\left(
\frac{1}{R}\Psi^{-1}(\sqrt{v}\,\CP_2(q))
\right),
\qquad q\in B(1).
$$
Then $\mathcal T_\omega$ is a diffeomorphism of $B(1)$ onto itself. Moreover,
for $q=z/R$,
$$
\RCtot(z)
=
\CP_2(\mathcal T_\omega(z/R))R
=
\Psi^{-1}(\sqrt{v}\,\CP_2(z/R)).
$$
Set $y=\sqrt{v}\,\CP_2(z/R)$. Then
$$
\begin{aligned}
\pi_B(z)
&=
\pi(\RCtot(z))|J_{\RCtot(z)}| 
=
\pi(\Psi^{-1}(y))
|J_{\Psi^{-1}}(y)|
\left|J_{\CP_2(z/R)R}\right| 
=
\pi_0(y)
\left|J_{\CP_2(z/R)R}\right|.
\end{aligned}
$$
 Since
$R^2=d+v$,
$$
\pi_B(z)
\propto
\left(1+|\CP_2(z/R)|^2\right)^{-(v+d)/2}
\left|J_{\CP_2(z/R)R}\right|.
$$
Using
$$
1+|\CP_2(z/R)|^2
=
\frac{1}{1-|z|^2/R^2},
\qquad
\left|J_{\CP_2(z/R)}\right|
=
\left(1-\frac{|z|^2}{R^2}\right)^{-(1+d/2)},
$$
we obtain
$$
\pi_B(z)
\propto
\left(1-\frac{|z|^2}{R^2}\right)^{v/2-1}.
$$
Therefore
$$
U_B(z)
=
C
-
\left(\frac{v}{2}-1\right)
\log\left(1-\frac{|z|^2}{R^2}\right),
$$
which is convex on $B(R)$ because $v\geq2$. Hence,
Assumption~\ref{assump:convex} holds. 

Note that $\pi_B$ is proportional to the target density in Example~\ref{ex:student_mixing}
with $R^2=d+v$. Thus Assumptions~\ref{assump:isotropic}  (with $C=1$)
and Assumption~\ref{assump:cold} (with parameters given in~\eqref{eq:Cold_params}) hold by the same arguments.
\end{proof}

\begin{rem}[Boundary regularity of the transport]
\label{rem:KR_boundary_regular}
The Knothe--Rosenblatt transport used in
Example~\ref{ex:multimodal_student_mixing} should be understood as an
open-domain transport. By itself, it does not guarantee that the $\mathcal{T}_\omega$  extends to a $C^1$-diffeomorphism of the closed ball
$\overline{B}(1)$. This boundary regularity is a separate condition.

The required closed-ball regularity may either be assumed directly, or obtained from compact-domain volume-form transport results. In particular,
Dacorogna--Moser type prescribed-Jacobian theorems, together with
support-control refinements, imply that if two smooth positive densities on
$B(R)$ have equal total mass and agree in a neighbourhood of
$\partial B(R)$, then the transport can be chosen to be the identity near
$\partial B(R)$~\cite[Thm~1]{teixeira2017dacorogna} (see also~\cite{dacorogna1990partial,teixeira2017dacorogna} for more details). Such a transport is a $C^1$-diffeomorphism of
$\overline{B}(R)$ with Jacobian bounded above and bounded away from zero.
\end{rem}

\section{Markov kernels for the Diffeomorphic Contraction Sampler}
\label{app:Algs}
In order to apply Algorithm~\ref{alg:DCS} in practice, it remains to specify the Markov kernel $Q$ on the ball $B(R)$ targeting the measure $\nu_B$ with density proportional to $\pi_B$ in~\eqref{eq:total_target}.
The fully preconditioned transformation $\RCtot$ provides a powerful mechanism for leveraging efficient, fast-mixing sampling algorithms originally designed for convex sets.  We now detail five well-studied algorithms for sampling from target laws on $B(R)$.

\subsection{Random Walk Metropolis (RWM)}
The first and most straightforward approach is the Random Walk Metropolis kernel with  Gaussian proposals. While it does not use the gradient information of the target, its each step is computationally inexpensive and highly robust. When restricted to $B(R)$, proposals that fall outside the ball are  rejected.

\begin{namedalgorithm}{DCS-RWM}{DCS with Gaussian Random Walk Metropolis}
\begin{algorithmic}[1]
\Statex \textbf{RWM kernel $Q$ (Line~\ref{alg:DCS:step2} of Algorithm~\ref{alg:DCS})}: requires target $\pi_B$ in~\eqref{eq:total_target} and $\sigma^2>0$
\State \textbf{Given:} current  state $\mathbf{z}$ in ball $B(R)$
\State Sample proposal $\hat{\mathbf{Y}} \sim \mathcal{N}(\mathbf{z}, \sigma^2 I_d)$
\If{$|\hat{\mathbf{Y}}| \ge R$}
    \State Set $\hat{\mathbf{Z}} = \mathbf{z}$ \Comment{Reject out-of-bounds proposals}
\Else
    \State With probability $1 \wedge \frac{\pi_{B}(\hat{\mathbf{Y}})}{\pi_{B}(\mathbf{z})}$, set $\hat{\mathbf{Z}} = \hat{\mathbf{Y}}$; otherwise, set $\hat{\mathbf{Z}} = \mathbf{z}$
\EndIf
\end{algorithmic}
\end{namedalgorithm}
The hyper-parameter $\sigma^2>0$, multiplying the identity matrix $I_d$ on $\R^d$, is specific to the kernel in~\ref{alg:DCS-RWM} and can be tuned using the standard criteria in~\cite{roberts1997weak}. 

\subsection{Hit-and-Run Sampling}
The second algorithm is the Hit-and-Run  sampler~\ref{alg:DCS-HnR}, which has been extensively studied due to its efficient exploration of convex bodies \cite{geman1984stochastic}. Most notably, the seminal work of Lov\'{a}sz and Vempala demonstrated its rapid mixing properties~\cite{lovasz2004hit, MR2309621}. Algorithm~\ref{alg:DCS-HnR} simulates a random direction $\theta$ in the unit sphere $\mathbb{S}^{d-1}\coloneqq \{x\in \R^d: |x|=1\}$  in $\R^d$ and then samples from the one-dimensional restriction of the target density along the chord in $B(R)$ in the direction $\theta$ (through the current point).

\begin{namedalgorithm}{DCS-HnR}{DCS with Hit-and-Run}
\begin{algorithmic}[1]
\Statex \textbf{HnR kernel $Q$ (Line~\ref{alg:DCS:step2} of Algorithm~\ref{alg:DCS})}: requires target $\pi_B$ in~\eqref{eq:total_target}
\State \textbf{Given:} current  state $\mathbf{z}$ in ball $B(R)$
\State Sample a random direction $\theta$ uniformly from the unit sphere $\mathbb{S}^{d-1}$
\State Determine the line segment $L = \{ t \in \R:\mathbf{z} + t\theta \in B(R)\}$
\State Sample $t^*$ from the 1D density proportional to $\pi_{B}(\mathbf{z} + t\theta)$ restricted to $L$
\State Set $\hat{\mathbf{Z}} = \mathbf{z} + t^*\theta$
\end{algorithmic}
\end{namedalgorithm}

Step 4 may be implemented using a valid univariate slice-sampling
transition targeting the density proportional to
$t\mapsto\pi_B(\mathbf z+t\theta)$ on the feasible chord. This
produces a different Markov kernel from ideal Hit-and-Run, which
samples exactly from the one-dimensional conditional distribution.
The slice-based kernel nevertheless leaves $\pi_B$ invariant and
therefore introduces no asymptotic bias. The Hit-and-Run convergence
bounds stated in this paper concern the ideal exact-conditional
kernel.

\subsection{Ball Walk}
The Ball-Walk algorithm~\ref{alg:DCS-BW} is conceptually similar to~\ref{alg:DCS-RWM}. Its proposals are drawn uniformly from a small local ball of radius $\gamma>0$ rather than a Gaussian distribution with variance $\sigma^2$.

\begin{namedalgorithm}{DCS-BW}{DCS with Ball Walk}
\begin{algorithmic}[1]
\Statex \textbf{BW kernel $Q$ (Line~\ref{alg:DCS:step2} of Algorithm~\ref{alg:DCS})}: requires target $\pi_B$ in~\eqref{eq:total_target} and $\gamma>0$
\State \textbf{Given:} current  state $\mathbf{z}$ in ball $B(R)$
\State Sample proposal $\hat{\mathbf{Y}}$ uniformly from $\mathbf{z} + \gamma B(1)$
\If{$|\hat{\mathbf{Y}}| \ge R$}
    \State Set $\hat{\mathbf{Z}} = \mathbf{z}$
\Else
    \State With probability $1 \wedge \frac{\pi_{B}(\hat{\mathbf{Y}})}{\pi_{B}(\mathbf{z})}$, set $\hat{\mathbf{Z}} = \hat{\mathbf{Y}}$; otherwise, set $\hat{\mathbf{Z}} = \mathbf{z}$
\EndIf
\end{algorithmic}
\end{namedalgorithm}

\subsection{Dikin Walk}
Improving the performance of Gaussian random walk~\ref{alg:DCS-RWM} and ball walk~\ref{alg:DCS-BW} requires the proposal distribution to adapt to the geometry of the underlying domain. The Dikin Walk Algorithm~\ref{alg:DCS-DW} generates proposals from the Dikin ellipsoid defined by the soft threshold or Lewis weights (see e.g.~\cite{jiang2024regularized,kook2024gaussian}), which led to recent breakthroughs demonstrating that these affine-invariant walks achieve remarkably fast sampling rates. At the heart of the algorithm is the bounded covariance function $H:B(R)\to \mathcal{S}_d^{+}$, mapping  $B(R)$ to a space of bounded positive semi-definite matrices $\mathcal{S}_d^{+}$. 

\begin{namedalgorithm}{DCS-DW}{DCS with Dikin Walk}
\begin{algorithmic}[1]
\Statex \textbf{DW kernel $Q$ (Line~\ref{alg:DCS:step2} of Algorithm~\ref{alg:DCS})}: requires target $\pi_B$ in~\eqref{eq:total_target}, a local metric $H:B(R)\to \mathcal{S}_d^{+}$, step-size $\gamma>0$, and covariance floor $\lambda\ge0$
\State \textbf{Given:} current  state $\mathbf{z}$ in ball $B(R)$
\State Draw $U \sim \text{Unif}(0, 1)$ 
\If{$U \ge 1/2$}
    \State Set $\hat{\mathbf{Z}} = \mathbf{z}$ \Comment{Lazification step: remain at current state with probability $1/2$}
\Else
\State Define local covariance: $\Sigma_\lambda(\mathbf z) =\frac{\gamma^2}{d}\left(H(\mathbf z)^{-1}+\lambda I_d\right).$

    \State Sample proposal $\hat{\mathbf{Y}}$ with density $\mathbf{y}\mapsto n(\mathbf{y};\mathbf{z}, \Sigma_\lambda(\mathbf{z}))$  \Comment{$n$ 
    denotes normal density on $\R^d$}
    \If{$|\hat{\mathbf{Y}}| \ge R$}
        \State Set $\hat{\mathbf{Z}} = \mathbf{z}$ \Comment{Reject out-of-bounds proposals}
    \Else
        \State Compute the acceptance probability 
        $\alpha = \min \left( 1, \frac{\pi_{B}(\hat{\mathbf{Y}})}{\pi_{B}(\mathbf{z})} \cdot \frac{n(\mathbf{z}; \hat{\mathbf{Y}}, \Sigma_\lambda(\hat{\mathbf{Y}}))}{n(\hat{\mathbf{Y}}; \mathbf{z}, \Sigma_\lambda(\mathbf{z}))} \right)$
        \State With probability $\alpha$, set $\hat{\mathbf{Z}} = \hat{\mathbf{Y}}$; otherwise, set $\hat{\mathbf{Z}} = \mathbf{z}$
    \EndIf
\EndIf
\end{algorithmic}
\end{namedalgorithm}
The key ingredient of Algorithm~\ref{alg:DCS-DW} is the choice of the local
covariance matrix $\Sigma_\lambda$. For the classical Dikin walk on
$B(R)$, one takes $\lambda=0$ and lets $H$ be the Hessian of the
logarithmic barrier:
\begin{equation}
\label{eq:local_metric}
H(z)=d\nabla^2 \varphi(z),\qquad\text{where}   \quad     \varphi(z)
        :=
        -\log(R^2-|z|^2)
\end{equation}
equals the logarithm of the distance of $z\in B(R)$ to the boundary of the ball $B(R)$, see~\cite{kook2024gaussian},
This choice gives fast mixing from a warm start, see Theorem~\ref{thm:dcs_bounds} above. However, since the Hessian of $\varphi$ degenerates near the boundary, the resulting chain may mix slowly when started from a poor initial distribution. To address this issue, we
use the covariance-floored Dikin proposal
\begin{equation}
\label{eq:regularised:Dikin_H}
        \Sigma_\lambda(z)
        :=
        \frac{\gamma^2}{d}
        \left(
        H(z)^{-1}+\lambda I_d
        \right),
        \qquad \lambda>0
\end{equation}
This is similar in spirit to the \textit{regularised Dikin walk} (see, e.g., \cite{jiang2024regularized}).

\subsection{Projected Langevin Monte Carlo}
Algorithm~\ref{alg:DCS-LMC}  incorporates the gradient information of the potential of $\pi_B$ in~\eqref{eq:total_target} to guide the proposal step \cite{MR3802303}. 
Let $\Pi_{\overline{B}(R)}:\R^d\to\overline{B}(R)$ denote the
Euclidean projection onto the closed ball,
\[
\Pi_{\overline{B}(R)}(y)
=
\begin{cases}
y, & |y|\leq R,\\[0.3em]
R\,y/|y|, & |y|>R.
\end{cases}
\]
Algorithm~\ref{alg:DCS-LMC} uses the projected Langevin update
of~\cite{MR3802303}. In contrast to rejection at the boundary,
an unconstrained proposal outside $B(R)$ is projected onto
$\partial B(R)$.

\begin{namedalgorithm}{DCS-LMC}{DCS with projected Langevin Monte Carlo}
\begin{algorithmic}[1]
\Statex \textbf{Projected-LMC kernel:}
requires the potential $U_B=-\log\pi_B$ and a step size $\eta>0$
\State \textbf{Given:} current state $\mathbf z\in\overline{B}(R)$
\State Sample $\xi\sim\mathcal N(0,I_d)$
\State Compute the unconstrained Langevin step
\[
\widetilde{\mathbf Y}
=
\mathbf z
-\frac{\eta}{2}\nabla U_B(\mathbf z)
+\sqrt{\eta}\,\xi
\]
\State Project the proposal onto the closed ball:
\[
\widehat{\mathbf Z}
=
\Pi_{\overline{B}(R)}(\widetilde{\mathbf Y})
\]
\end{algorithmic}
\end{namedalgorithm}

Projected LMC is an unadjusted finite-step sampling algorithm.
Its transition kernel need not have $\nu_B$ as its invariant
distribution. The result below therefore controls the
finite-time error
$\|\mathcal L(Z_N)-\nu_B\|_{\mathrm{TV}}$
rather than the mixing time to an invariant distribution.

\section{Variational optimisation of the automorphism of the ball}
\label{app:mobius_vi}

To formulate the variational objective using a parameter-independent
base distribution, let
$u_1(s)
=
\mathbf 1_{B(1)}(s)/\operatorname{vol}_d(B(1))$
denote the density of the uniform distribution on the unit ball $B(1)$ in $\R^d$
and let $S\sim u_1$. For
$\theta=(\mu,R,\omega)$, define
\begin{equation}
\label{eq:unit_ball_map}
F_\theta:B(1)\to\R^d,
\qquad
F_\theta(s)
:=
\CP_{\mathrm{total}}(Rs)
=
\mu+R\cdot\CP_\beta(\mathcal T_\omega(s)).
\end{equation}
Thus $F_\theta$ is the unit-ball parametrization of the total
diffeomorphism. Let
$q_\theta=(F_\theta)_*u_1$
be the push-forward of the fixed base distribution under
$F_\theta$.

For $x=F_\theta(s)$, the change-of-variables formula gives
$\log q_\theta(x)
=
\log u_1(s)
-
\log\left|\det DF_\theta(s)\right|$.
Let $\overline\pi$ denote the normalized target density
proportional to $\pi$. Then
\begin{align*}
D_{\mathrm{KL}}(q_\theta\|\overline\pi)
&=
\mathbb E_{S\sim u_1}
\left[
\log u_1(S)
-
\log\left|\det DF_\theta(S)\right|
-
\log\overline\pi(F_\theta(S))
\right].
\end{align*}
The terms $\log u_1(S)=-\log\operatorname{Vol}(B(1))$ and the
normalizing constant of $\pi$ do not depend on $\theta$.
Consequently, the variational parameters can be obtained by
minimizing
\begin{equation}
\label{eq:fixed_base_vi}
\theta^\star
\in
\arg\min_\theta
\mathcal E_{\mathrm{VI}}(\theta),
\qquad
\mathcal E_{\mathrm{VI}}(\theta)
:=
\mathbb E_{S\sim u_1}
\left[
-\log\pi(F_\theta(S))
-
\log\left|\det DF_\theta(S)\right|
\right].
\end{equation}

The Jacobian in~\eqref{eq:fixed_base_vi} decomposes as
$\left|\det DF_\theta(s)\right|
=
R^d
J_{\CP_\beta}(\mathcal T_\omega(s))
\cdot J_{\mathcal T_\omega}(s)$,
implying
\[
\log\left|\det DF_\theta(s)\right|
=
d\log R
+
\log J_{\CP_\beta}(\mathcal T_\omega(s))
+
\log J_{\mathcal T_\omega}(s).
\]
Equivalently, define the normalized pull-back density on $B(1)$ by
\[
\widetilde\pi_\theta(s)
=
\overline\pi(F_\theta(s))
\left|\det DF_\theta(s)\right|.
\]
Invariance of relative entropy under the diffeomorphism
$F_\theta$ gives
\[
D_{\mathrm{KL}}(q_\theta\|\overline\pi)
=
D_{\mathrm{KL}}(u_1\|\widetilde\pi_\theta).
\]
Thus the objective encourages the pull-back target on the unit ball
to resemble the uniform distribution in reverse KL. Whether this
also produces the log-concavity, smoothness, or isotropy required
by the non-asymptotic theory depends on the expressiveness of the
chosen family $\mathcal T_\omega$.
Crucially, the distribution of $S$ is independent of $\theta$.
Consequently, Monte Carlo estimates of~\eqref{eq:fixed_base_vi}
can be differentiated pathwise with respect to $\mu$, $R$, and
$\omega$ and optimized using Adam.

The following subsections specialise the variational objective in~\eqref{eq:fixed_base_vi} to the M\"obius automorphism $\mathcal M_\delta$ in~\eqref{eq:mobius_unit}, used in
our numerical experiments. In this setting,
$\omega=\delta\in B(1)$,
$\mathcal T_\omega=\mathcal M_\delta$
and the constrained transformation parameters are
$\theta=(\mu,R,\delta)
\in
\R^d\times(0,\infty)\times B(1)$.
Recall from~\eqref{eq:unit_ball_map} that the corresponding
unit-ball parametrisation is
$F_\theta(s)
=
\mu+R\cdot\CP_\beta(\mathcal M_\delta(s))$,
$s\in B(1)$.

\subsection{Jacobian of the M\"obius parametrisation}

For $s,\delta\in B(1)$, define
$a_\delta(s)
:=
1+2\langle s,\delta\rangle+|\delta|^2|s|^2$.
The M\"obius transformation in~\eqref{eq:mobius_unit} satisfies
$\left|\det D\mathcal M_\delta(s)\right|
=
\left(
(1-|\delta|^2)/(a_\delta(s))
\right)^d$.
For notational convenience, define the log-Jacobian of the
unit-ball map by
$\Lambda_\theta(s)
:=
\log\left|\det DF_\theta(s)\right|$.
The chain rule and the Jacobian formula for $\CP_\beta$ give
\begin{equation}
\label{eq:mobius_total_log_jacobian}
\begin{aligned}
\Lambda_\theta(s)
={}&
d\log R
+d\log(1-|\delta|^2)
-d\log a_\delta(s)
-\left(1+\frac{d}{\beta}\right)
\log\left(
1-|\mathcal M_\delta(s)|^\beta
\right).
\end{aligned}
\end{equation}

Consequently, the single-sample contribution to the fixed-base
objective~\eqref{eq:fixed_base_vi} is
\begin{equation*}
\ell_{\mathrm{VI}}(\theta;s)
:=
-\log\pi(F_\theta(s))
-\Lambda_\theta(s),
\>s\in B(1),\quad\text{implying $\mathcal E_{\mathrm{VI}}(\theta)
=
\mathbb E_{S\sim u_1}
\left[
\ell_{\mathrm{VI}}(\theta;S)
\right]$.}
\end{equation*}

\subsection{Unconstrained parametrisation and stochastic optimisation}
\label{app:opt_alg}

Adam~\cite{kingma2015adam} operates on unconstrained Euclidean parameters. We 
introduce
$\vartheta=(\mu,\zeta,v)
\in
\R^d\times\R\times\R^d$, where $R(\zeta)=\exp(\zeta)$.
To map $v\in\R^d$ smoothly into the open unit ball, define $\psi(v):=\tanh(|v|)v/|v|$ for $v\neq0$ and $\psi(0)=0$.
Then
$|\psi(v)|=\tanh(|v|)<1,$
so that $\delta=\psi(v)\in B(1)$. The constrained
parameters associated with $\vartheta$ are thus
$\theta(\vartheta)
:=
\bigl(\mu,\exp(\zeta),\psi(v)\bigr)$.
For a mini-batch
$s^{(1)},\ldots,s^{(M_{\mathrm{VI}})}\in B(1)$, set
\begin{equation*}
\widehat{\mathcal E}_{\mathrm{VI}}
\left(
\vartheta;
s^{(1)},\ldots,s^{(M_{\mathrm{VI}})}
\right)
:=
\frac{1}{M_{\mathrm{VI}}}
\sum_{j=1}^{M_{\mathrm{VI}}}
\ell_{\mathrm{VI}}
\left(
\theta(\vartheta);s^{(j)}
\right).
\end{equation*}

\begin{namedalgorithm}{VI}
{Variational optimisation of the M\"obius preconditioner}
\begin{algorithmic}[1]
\Require Unnormalised target log-density
$x\mapsto\log\pi(x)$; radial exponent $\beta>0$;
number of optimisation steps $K_{\mathrm{VI}}$;
mini-batch size $M_{\mathrm{VI}}$;
Adam learning rate $\alpha_{\mathrm{VI}}$;
initial unconstrained parameters
$\vartheta_0=(\mu_0,\zeta_0,v_0)$
\Ensure Optimised constrained parameters
$\theta^\star=(\mu^\star,R^\star,\delta^\star)$

\For{$k=0,\ldots,K_{\mathrm{VI}}-1$}
    \State Draw
    $s_k^{(1)},\ldots,s_k^{(M_{\mathrm{VI}})}$
    independently from $\operatorname{Unif}(B(1))$
    \Comment{Fixed base distribution}

    \State Compute the pathwise gradient
    \[
    g_k
    \gets
    \nabla_\vartheta
    \widehat{\mathcal E}_{\mathrm{VI}}
    \left(
    \vartheta;
    s_k^{(1)},\ldots,s_k^{(M_{\mathrm{VI}})}
    \right)
    \bigg|_{\vartheta=\vartheta_k}
    \]

    \State Update
    $\vartheta_{k+1}
    \gets
    \operatorname{AdamStep}
    \left(
    \vartheta_k,g_k;\alpha_{\mathrm{VI}}
    \right)$
    \Comment{Algorithm in~\cite{kingma2015adam}}
\EndFor

\State Set
$\theta^\star
\gets
\theta(\vartheta_{K_{\mathrm{VI}}})$

\State \textbf{return}
$\theta^\star=(\mu^\star,R^\star,\delta^\star)$
\end{algorithmic}
\end{namedalgorithm}

As the base samples are drawn from a distribution that does
not depend on $\vartheta$, the gradient in
Algorithm~\ref{alg:VI} is a standard pathwise gradient
estimator. The normalization constant of $\pi$ is not required.

\end{document}